\documentclass[journal,onecolumn]{IEEEtran}
\usepackage{amsmath,amsfonts,amssymb,amsthm,hyperref,enumerate,multicol,bm}
\usepackage{calligra,indentfirst,epsfig,lettrine}
\usepackage{caption}
\usepackage{relsize}
\usepackage{tabls}
\usepackage{array}
\usepackage{longtable}
\calligra

\DeclareFontFamily{U}{mathc}{}
\DeclareFontShape{U}{mathc}{m}{it}%
{<->s*[1.03] mathc10}{}

\DeclareMathAlphabet{\mathscr}{U}{mathc}{m}{it}

\usepackage{comment}
\makeatletter
\newcommand*{\rom}[1]{\expandafter\@slowromancap\romannumeral #1@}
\makeatother
\usepackage{mathtools}
\usepackage{graphicx}
\usepackage{subfigure}
\usepackage{xcolor}
\usepackage{color, soul}
\usepackage{mathrsfs}
\DeclareMathAlphabet{\mathpzc}{OT1}{pzc}{m}{it}

\usepackage{ upgreek }
\newcommand{\etal}{\textit{et al.}}

\newtheorem{definition}{Definition}
\newtheorem{corollary}{Corollary}
\newtheorem{theorem}{Theorem}
\newtheorem{example}{Example}
\newtheorem{lemma}{Lemma}
\newtheorem{remark}{Remark}
\newtheorem{proposition}{Proposition}
\allowdisplaybreaks
\usepackage[table]{xcolor}

\begin{document}
\title{{Maximal achievable service rates of some classes of linear codes}}

\author {Priyanka~Choudhary, Monika Yadav
       and~Maheshanand~Bhaintwal\\ 
       % Department of Mathematics, Indian Institute of Technology Roorkee, Roorkee 247667, India
       
       \thanks{The authors are with the Department of Mathematics, Indian Institute of Technology Roorkee, Roorkee 247667, India (e-mails: \href{mailto:priyanka_c@ma.iitr.ac.in}{priyanka\_c@ma.iitr.ac.in}, \href{mailto:monikay@iiitd.ac.in}{monikay@iiitd.ac.in},
   \href{mailto:maheshanand@ma.iitr.ac.in}{maheshanand@ma.iitr.ac.in}).}}

\date{}
\maketitle
\begin{abstract}
In this paper, we investigate lower bounds on the maximum achievable service rates
for data symbols in certain classes of linear codes, including cyclic codes and low-density parity-check (LDPC) codes, that are derived from combinatorial structures such as $t$-designs, difference sets, and balanced incomplete block designs (BIBDs). 
We first establish a lower bound on the maximum achievable service rate for each data symbol in the following two cases: (i) systematic linear codes $C$ under the assumption that the supports of codewords of a fixed weight in $C^\perp$ form a BIBD, and (ii) non-systematic binary codes under the assumption that the supports of codewords of a fixed weight in $C^\perp$ form a  
$t$-design.
We then investigate the linear codes obtained from the incidence matrices of BIBDs, particularly certain classes of BIBD-LDPC codes,   and show how the parameters of the underlying designs can be exploited to determine lower bounds on the maximum achievable service rates for the data symbols of the corresponding linear code. In addition, we analyze the maximum achievable service rates of systematic extended linear codes. We show that the existence of a symmetric BIBD (SBIBD) corresponding to a dual codeword can be used to derive a lower bound on the
maximum achievable service rate of the associated systematic cyclic code. We also present some families of cyclic codes constructed from difference sets and obtain explicit lower bounds on the maximum achievable service rates for their data symbols. In particular, we determine the exact values of the maximum achievable service rates
for each data symbol of cyclic codes arising from Singer difference sets.

\end{abstract}


%Service rate region; Extended codes; Cylic codes; LDPC codes; Symmetric BIBD.\\
% {\bf 2020 Mathematics Subject Classification}:  94B15. 

\IEEEpeerreviewmaketitle

\section{Introduction}

\IEEEPARstart{T}o meet the increasing user demands, modern large-scale distributed storage systems require data-access mechanisms that are both reliable and scalable. Traditionally, replication-based storage schemes have been employed to ensure data reliability and availability \cite{Aktas}.   However, the increasing storage overhead associated with replication has motivated the use of erasure coding techniques, which provide higher reliability without sacrificing storage efficiency. 
 Moreover, in practical systems, the popularity or demand of different data objects can be highly nonuniform. In such scenarios, hybrid redundancy schemes that combine replication and coding have attracted considerable attention   \cite{Dimakis,Gardner,Gopalan,Huang,Joshi,Liang}.   To assess the capability of these coded distributed storage systems to handle user requests efficiently, the idea of the service rate region (SRR) was introduced \cite{Noori}. The SRR represents the set of all request-rate vectors that can be supported simultaneously by a storage system, subject to the server-capacity constraints. This concept is closely connected to the problems of load balancing and access balancing in distributed storage systems
\cite{MAktas,Aktas,Noori}.
 
%The rapid growth of large-scale distributed storage systems has intensified the need for reliable, efficient, and scalable mechanisms for data access. 

The service rate region provides a quantitative metric for assessing the request-serving capacity of a storage scheme. By delineating all demand vectors that can be supported simultaneously, it facilitates a direct comparison among various storage schemes with an identical number of storage nodes. In particular, if two schemes utilize identical storage resources, the one with the larger service rate region can support a broader spectrum of request arrival rates and therefore offers higher throughput and improved load balancing \cite{MAktas,Aktas}. Thus, the service rate region functions as a performance metric for evaluating the efficiency of distributed coded storage schemes.

Over the past decade, the service rate regions of various classes of linear codes have been investigated. In particular, explicit characterizations and bounds have been obtained for several important code families, including MDS codes, Simplex codes, Reed-Muller codes, locally recoverable codes, and Hamming codes \cite{Aktas,LRC,HamPC,Combit,Geo,LyH,LyHS}. Despite these advances, the service rate regions of many important classes of linear codes have not been studied. In particular, the SRRs of extended linear codes and cyclic codes remains largely unexplored.  More generally, determining the relationship between the algebraic and combinatorial structures of a code and the service capabilities of the corresponding distributed storage system remains a significant open problem.

Recently, Giusto {\etal} \cite{oval} investigated the connection of ovals in projective geometry to service rate regions. They presented a construction of a non-systematic generator matrix of an MDS code of dimension $3$ over $\mathbb{F}_q$ that has a larger service rate region than the systematic generator matrix of the same code. It may be noted here that not every generator matrix of a code yields the same SRR. This also highlights the need to find the generator matrix that gives the optimal SRR. In a related work, Kazemi {\etal} \cite{Combit} modeled recovery sets as edges of a graph and obtained bounds on SRR by showing that allocating the service rates to recovery sets is the same as the fractional matching problem in graphs. 
%This was an approach to investigating SRR through a combinatorial framework. 
Further,  Kazemi {\etal} \cite{Geo} approached the SRR problem from a finite-geometric perspective, using the geometric structure of Reed–Muller codes to characterize their recovery sets and service rate regions. Since determining the exact service polytope is generally challenging, especially for long-length codes, bounds on the maximum achievable service rate for a data symbol provide a fundamental limit for service rates of codes, which also assists in evaluating performance efficiency.
 Recently, Ly and Soljanin \cite{Ly}  provided a bound on the maximum achievable service rates for data symbols of a linear code by exploiting the symmetry, regularity, and incidence properties of combinatorial designs.
 %By exploiting the symmetry, regularity, and incidence properties of combinatorial design, a bound on maximum achievable service rates of data symbols of a linear code was obtained in \cite{Ly}. 
 %These works demonstrate that the structure and interaction of recovery sets can be effectively analyzed using tools from combinatorics, finite geometry, and design theory. 
  In distributed storage systems, establishing service rate bounds is essential because they determine how many simultaneous user requests can be fulfilled without overloading any node, thereby 
maximizing resource utilization and system efficiency. %Understanding these bounds enables system designers to make informed trade-offs between code properties—such as redundancy levels, repair bandwidth, and update complexity—and the achievable service rate. Furthermore, service rate bounds provide a benchmark for evaluating practical code constructions, revealing whether existing codes approach theoretical limits or whether better designs can improve concurrent access throughput. 
In real-world applications such as cloud storage, video streaming, and large-scale data centers where many users access data at the same time, the difference between actual service rates and their theoretical limits can cause system bottlenecks and higher delays. Therefore, setting precise limits on the maximum service rate helps guide the development of better codes and capacity planning. 
%This makes it essential for designing efficient, scalable distributed systems that can perform reliably under real-world demands. 

 It is well known that cyclic codes form a fundamental class of linear codes with rich algebraic structure. The cyclic nature of these codes provides significant computational advantages,  as both  encoding and decoding can be efficiently implemented using linear feedback shift registers  (refer to \cite[Sec. 4.2]{Huffman}).  Furthermore, algebraic techniques from polynomial algebra and number theory can be employed for the design and analysis of these codes. Cyclic codes also have strong connections to finite geometry, combinatorial designs, and difference sets \cite{Graham,Weldon}. Several important cyclic codes arise from Singer and quadratic residue difference sets, while the Golay codes are closely linked to the Witt designs \cite{Huffman}.  Weldon \cite{Weldon} showed that the class of cyclic codes arising from Singer difference sets is almost as powerful as BCH codes and considerably simpler to implement. 
 % Difference sets can be used directly to construct the defining set of a cyclic code, yielding codes with known weight distributions and distance properties determined by the combinatorial parameters of the underlying difference set. 
  Another important class of linear codes is the class of regular Low-Density Parity-Check (LDPC) codes, which were originally introduced by Gallager in 1962 \cite{GalH}.
One of the primary advantages of these regular LDPC codes is that they can be encoded and decoded with very low complexity. Moreover, experimental results demonstrate that these codes perform well with iterative decoding \cite{Vasic}. Their performance over classical codes on various channel models, such as AWGN and binary symmetric channels, has established them as a favored choice for next-generation communication and storage systems \cite{Lan}.

% making regular LDPC codes a good balance between theoretical optimality and ease of practical implementation (BIBD-LDPC codes admit low hardware complexity implementation).

Motivated by these developments, this paper investigates bounds on the maximum achievable service rates for data symbols of extended linear codes, cyclic codes, and LDPC codes that are associated with combinatorial structures such as $t$-designs, difference sets, and  balanced incomplete block designs (BIBDs).

 %to explore other combinatorial structures that exhibit highly regular incidence patterns. In particular, cyclic designs, difference sets, symmetric balanced incomplete block designs (SBIBDs), and LDPC codes derived from such designs exhibit rich combinatorial properties, including regulated block intersections, cyclic symmetry, and numerous/disjoint recovery sets. These characteristics make them promising candidates for investigating service rate regions and for constructing coded storage schemes with the service capabilities known to some extent. 
 
 This paper is organized as follows: In Section \ref{prelim}, we introduce some preliminaries essential to establishing our main results. In Section \ref{linear}, %we investigate the maximal achievable service rates of linear codes associated with BIBDs. 
we establish a lower bound on the maximum achievable service rate for each data symbol of a systematic linear code when
the supports of codewords of a fixed weight in the dual code form a BIBD (Theorem \ref{thm2.3}). We also provide a lower bound on the maximum achievable service rate for a data symbol of a non-systematic binary linear code when
the supports of codewords of a fixed weight in the dual code form a $t$-design (Theorem \ref{nonsymetriclinear}). 
% Further, we derive a lower bound on the maximum achievable service rate for a data symbol of a non-systematic binary Hamming code, provided that there exists a recovery set of a specified size for that data symbol (Theorem \ref{thm2.5}).
  In Section \ref{incidence}, we present linear codes constructed from the incidence matrices of BIBDs, particularly certain classes of BIBD-LDPC codes, and show how the parameters of the underlying design can be used to derive lower bounds on the maximum achievable service rates for  the data symbols of the corresponding linear code  (Theorem \ref{LDPC}). In Section \ref{extended}, we explore the maximum achievable service rates of systematic extended linear codes and explore the conditions under which the service rates for the data symbols of the extended code exceed that of the original code (Theorems \ref{thm3.5}, \ref{thm2}, and \ref{thm3.4}). In Section \ref{cyclicsection}, we show how the existence of an SBIBD corresponding to a codeword in the dual code
can be used to derive a lower bound on the maximum achievable service rate of the associated systematic cyclic code  (Theorem \ref{thm4.1}). We also present some families of cyclic codes constructed from difference sets and obtain explicit lower bounds on the maximum achievable service rates for their data symbols. In particular, we determine exact values of the maximum achievable service rates for each data symbol of cyclic codes arising from Singer difference sets. These codes form an important class of cyclic codes that are nearly as powerful as BCH codes while being considerably simpler to implement \cite{Weldon}. In Section \ref{conclusion}, we mention a brief conclusion and discuss some future research directions.

 \section{Some preliminaries}\label{prelim} 

 In this section, we introduce the notation and terminology used throughout the paper. We begin with an overview of linear codes over finite fields, their generator matrices, and dual codes. We then recall the definitions of coded distributed systems and their associated service rate regions. Finally, we state the necessary concepts from design theory, including $t$-designs and balanced incomplete block designs (BIBDs), that will be needed to derive our main results.

 \subsection{Linear codes over finite fields}

 Let $q$ be a prime power, and let $\mathbb{F}_q$ represent the finite field of order $q.$  For a positive integer $n,$ let $\mathbb{F}_q^n$ denote the set of all $n$-tuples over $\mathbb{F}_q,$ which forms an $n$-dimensional vector space over $\mathbb{F}_q.$
A linear code $\mathcal{C}$ of length $n$ and dimension $k$ over $\mathbb{F}_q$ is defined as a $k$-dimensional subspace of $\mathbb{F}_q^n.$ The elements of $\mathcal{C}$  are called codewords.  A  $k\times n$ matrix $\mathcal{G}$ over $\mathbb{F}_{q}$ whose rows form a basis of $\mathcal{C}$ over $\mathbb{F}_{q},$ is called a generator matrix of $\mathcal{C}.$ Given an information vector $\mathbf{x}=(x_1,x_2,\ldots,x_k)\in\mathbb{F}_q^k,$ the corresponding codeword is obtained as $\mathbf{c}=\mathbf{x}\cdot\mathcal{G}\in\mathbb{F}_q^n.$ Thus, the code $\mathcal{C}$ can be characterized by the generator matrix $\mathcal{G}.$ 
The Hamming distance of the code $\mathcal{C}$, denoted by $d_H(\mathcal{C}),$ is the minimum Hamming distance between any two distinct codewords of $\mathcal{C}$. Since $\mathcal{C}$ is linear,
$
d_H(\mathcal{C})=\min\{w_H(\mathbf{c})~|~ \mathbf{c}(\neq \mathbf{0}) \in \mathcal{C}\},
$
where $w_H(\mathbf{c})$  denotes the Hamming weight of the codeword $\mathbf{c}.$ 
For a codeword $\mathbf{c}=(c_1,c_2,\ldots,c_n)\in \mathcal{C},$ the support of $\mathbf{c},$ denoted by $\mathcal{S}upp(\mathbf{c}),$ is defined as the set of coordinate positions at which $\mathbf{c}$  has nonzero entries. That is,
$\mathcal{S}upp(\mathbf{c})=\{i ~|~ c_i\neq 0 \text{ for } i \in [n]\}$  (throughout this paper,  $[\ell]$ denotes the set $\{1,2,\ldots,\ell\}$ for any positive integer $\ell$).
A linear code of length $n$ and dimension $k$  over $\mathbb{F}_q$ is referred to as an $[n,k]$ linear code over $\mathbb{F}_q.$ For $0\leq w \leq n,$ the number of codewords of weight $w$ in $\mathcal{C}$ is denoted by $A_w.$ 

The dual code of $\mathcal{C}$, denoted by $\mathcal{C}^\perp, $  is defined as the orthogonal complement of $\mathcal{C}$ in $\mathbb{F}_q^n,$ i.e.,
$$
\mathcal{C}^\perp=\{\mathbf{v}\in\mathbb{F}_q^n ~|~ \mathbf{v}\cdot \mathbf{c} = 0 \text{ for all } \mathbf{c}\in \mathcal{C}\},
$$
where $\mathbf{v} \cdot \mathbf{c}$  denotes the Euclidean inner product of $\mathbf{v}$  and $\mathbf{c}.$  Note that $\mathcal{C}^\perp$ is also a linear code with parameters $[n,n-k].$ Throughout this paper,  $d^\perp$ denotes the minimum Hamming distance of the dual code $\mathcal{C}^\perp.$ 
%\textcolor{red}{A generator matrix of  $\mathcal{C}^\perp$ is called a parity-check matrix of $\mathcal{C}.$} 
We will denote  the number  of codewords of weight $w$ in $\mathcal{C}^\perp$  by $A_w^\perp,$ where $0 \leq w \leq n.$

 \subsection{Service rate region of coded distributed storage systems}

In this section, we recall the basic concepts of coded distributed storage systems and their associated service rate regions. Throughout this section, let $n$ and $k$ be positive integers satisfying $2\leq k<n.$
Now fix a matrix $\mathcal{G}\in\mathbb{F}_q^{k\times n}$ of rank $k,$ and let $\mathcal{C}$ be the $[n,k]$ linear code  over $\mathbb{F}_q$ with $\mathcal{G}$ as a generator matrix.  We assume that $\mathcal{G}$  has no zero columns and denote its $j$-th column by $\mathcal{G}^{(j)}.$
A coded distributed system consists of $k$ data objects from $\mathbb{F}_q$ that are linearly encoded into $n$ coded objects over $\mathbb{F}_q$ and distributed (stored) across $n$ servers, with each server storing exactly one coded object. The system is completely characterized by the matrix $\mathcal{G}$, which is also called the generator matrix of the system. For a data vector
$\mathbf{x}=
(x_1,x_2,\ldots,x_k)\in\mathbb{F}_q^k,$ the corresponding encoded vector is
$\mathbf{x}\cdot\mathcal{G}\in\mathbb{F}_q^n.$
The $j$-th server stores the $j$-th coordinate of $\mathbf{x}\cdot\mathcal{G}.$ If the first $k$ columns of $\mathcal{G}$ form the $k\times k$ identity matrix, then the matrix $\mathcal{G}$  is said to be \emph{systematic}  and the corresponding linear code with generator matrix $\mathcal{G}$ is said to be a systematic linear code. We consider request vectors
$
(\lambda_1,\lambda_2,\ldots,\lambda_k)\in\mathbb{R}_{\ge 0}^{k},
$
where $\lambda_i$ denotes the request rate for the $i$-th data object (throughout this paper, $\mathbb{R}_{\ge 0}^{k}$ denotes the set of all $k$-tuples of nonnegative real numbers). Each server can process at most $\alpha$ requests per unit of time, referred to as the \emph{server capacity}. To serve a request for a given object, a user is assigned a collection of servers whose stored coded symbols allow recovery of that object. 

\begin{definition}
    Let $e_i$ denote the $i$-th standard basis vector of $\mathbb{F}_q^k.$ A set $R\subseteq [n]$ is called a recovery set for the $i$-th data object if
$e_i \in Span\{ \mathcal{G}^{(j)} ~|~ j\in R\},$
where $Span\{ \mathcal{G}^{(j)} ~|~ j\in R\}$ denotes the $\mathbb{F}_q$-linear span of the columns of $\mathcal{G}$ indexed by $R.$

\end{definition}
 A data object may be recoverable from multiple subsets of servers. The set of all recovery sets for the $i$-th data object is denoted by $\mathcal{R}_i^{\mathrm{all}}(\mathcal{G}).$ Furthermore, a set $R \in \mathcal{R}_i^{\mathrm{all}}(\mathcal{G})$ is called a minimal recovery set for the $i$-th data object if no other recovery set $R' \in \mathcal{R}_i^{\mathrm{all}}(\mathcal{G})$ satisfies the condition
$R' \subsetneq R.$ The set of all minimal recovery sets for the $i$-th data object is denoted by $\mathcal{R}_i^{\mathrm{min}}(\mathcal{G}).$ 

The following proposition provides a characterization of recovery sets for a data object in  a systematic linear code in terms of the supports of codewords in its dual code.
\begin{proposition}\label{Prop2.1}\cite[Prop. IV.4]{Alfarano}
  Let $\mathcal{C}$ be an $[n,k]$  linear code   over $\mathbb{F}_q$ with a systematic generator matrix.  
  % \textcolor{blue}{Let $\mathcal{G}$  be a systematic generator matrix of a coded distributed system.} 
   For $i\in[k]$ and $R\subseteq[n],$ the set $R$ is a recovery set for the $i$-th data object
if and only if either $R=\{i\}$ or there exists a codeword $\mathbf{d}\in \mathcal{C}^\perp$ such that
$
i\in \mathcal{S}upp(\mathbf{d})$ 
and $\mathcal{S}upp(\mathbf{d})\subseteq R\cup\{i\}.$

\end{proposition} 

\begin{definition}
 A recovery $\mathcal{G}$-system  is defined as a $k$-tuple $\mathcal{R}(\mathcal{G})=\big(\mathcal{R}_1(\mathcal{G}),\mathcal{R}_2(\mathcal{G}),\ldots,\mathcal{R}_k(\mathcal{G})\big)$ of sets of subsets of $[n],$ 
where each component $\mathcal{R}_i(\mathcal{G})$ is a nonempty subset of $\mathcal{R}_i^{\mathrm{all}}(\mathcal{G})$ for $i \in [k].$
   
\end{definition}

\begin{definition}\cite{alfarano2024service}
  Let $\mathcal{R}(\mathcal{G})=\big(\mathcal{R}_1(\mathcal{G}),\mathcal{R}_2(\mathcal{G}),\ldots,\mathcal{R}_k(\mathcal{G})\big)$ be a recovery $\mathcal{G}$-system, and let $\alpha$ be the server capacity. The \textbf{service rate region (SRR)} associated with $\mathcal{R}(\mathcal{G}), $ and denoted by $\Lambda\big(\mathcal{R}(\mathcal{G}),\alpha\big),$   is the set of all request-rate vectors
$(\lambda_1,\lambda_2,\ldots,\lambda_k)\in \mathbb{R}_{\geq 0}^{k}$
for which there exists a collection of nonnegative real numbers
$
\{\lambda_{i,R} ~|~ R\in \mathcal{R}_i(\mathcal{G})\}$  ($\lambda_{i,R}$ represent the fraction of the request rate $\lambda_i$ assigned to the recovery set $R$ for $i\in[k]$) satisfying  the following conditions:
\begin{equation}\label{eq1}
    \sum_{R\in \mathcal{R}_i(\mathcal{G})}\lambda_{i,R}=\lambda_i
\text{ for all }~ i\in [k],
\end{equation}
and \begin{equation}\label{eq2}
    \sum_{i=1}^{k}\sum_{\substack{R\in \mathcal{R}_i(\mathcal{G})\\ \ v\in R}}
\lambda_{i,R}
\le \alpha
\text{ for all }~ v\in [n].
\end{equation}
\end{definition}

Condition \eqref{eq1} ensures that the aggregate demand for each data object is completely satisfied across its available recovery sets, while condition \eqref{eq2} ensures that the total load allocated to each server does not exceed its service capacity $\alpha.$
By Remark II.6 of Alfarano {\etal} \cite{Alfarano}, we note that
$\Lambda(\mathcal{R}^{\mathrm{all}}(\mathcal{G}),\alpha)=\alpha\Lambda(\mathcal{R}^{\mathrm{all}}(\mathcal{G}),1),
$  where $\mathcal{R}^{\mathrm{all}}(\mathcal{G})= (\mathcal{R}_1^{\mathrm{all}}(\mathcal{G}),\mathcal{R}_2^{\mathrm{all}}(\mathcal{G}),\ldots,\mathcal{R}_k^{\mathrm{all}}(\mathcal{G})).$ Hence, 
it suffices to study the  service rate region $\Lambda(\mathcal{R}^{\mathrm{all}}(\mathcal{G}),1)$. Further, by Proposition III.3 of Alfarano {\etal} \cite{Alfarano}, we see that $\Lambda(\mathcal{R}^{\mathrm{all}}(\mathcal{G}),1)=\Lambda(\mathcal{R}^{\mathrm{min}}(\mathcal{G}),1),
$  where  $\mathcal{R}^{\mathrm{min}}(\mathcal{G})= (\mathcal{R}_1^{\mathrm{min}}(\mathcal{G}),\mathcal{R}_2^{\mathrm{min}}(\mathcal{G}),\ldots,\mathcal{R}_k^{\mathrm{min}}(\mathcal{G})).$ In view of this, throughout the paper, we focus on $\Lambda(\mathcal{R}^{\mathrm{min}}(\mathcal{G}),1)$, 

For each $i \in [k],$ let us define the coordinate-wise maximum
$$\lambda_{i}^{*}=\max\{ \lambda_i ~|~ \lambda= (\lambda_1,\lambda_2,\ldots,\lambda_k)\in\Lambda(\mathcal{R}^{\mathrm{min}}(\mathcal{G}),1)\}.$$
In other words, $\lambda_i^{*}$ is the largest request rate that can be supported for the $i$-th data object while ensuring that the resulting request-rate vector remains within the respective SRR, and it is called the \textit{maximum achievable rate} for the $i$-th data object.
%Equivalently, it is the maximum value attained by the $i$-th coordinate over all points in $\Lambda(\mathcal{R}(\mathcal{G}),1).$

\subsection{ $t$-Designs and BIBDs}
In this subsection, we recall the basic definitions and properties of $t$-designs, focusing on $2$-designs, which are recognized as balanced incomplete block designs (BIBDs). 
\begin{definition}\cite[Sec. 8.1]{Huffman}
  Let $n,\beta,t $ be positive integers such that $1\leq \beta<n.$  A $t-(n,\beta,\mu)$ design is a pair $\mathcal{D}=(V,\mathcal{X}),$ where $V$ is a set of $n$  points and $\mathcal{X}$ is a collection of $\beta$-subsets of $V,$ called blocks, with the property that every $t$-subset of $V$ is contained in exactly $\mu$ blocks. In particular, 
when $t=2,$ a $t$-design reduces to a
$2-(n,\beta,\mu)$ 
design, commonly known as a balanced incomplete block design (BIBD).  Furthermore, a BIBD is said to be a symmetric BIBD (SBIBD) if $\mathcal{D}$ contains exactly $n$  blocks.
\end{definition}

\begin{lemma}\label{lem2.1}\cite[Th. 8.1.3]{Huffman}
Let $\mathcal{D}=(V, \mathcal{X})$ be a $t-(n,\beta,\mu)$ design. Then, for every integer $s$ satisfying $1\le s\le t,$ $\mathcal{D}$ is also an $s-(n,\beta,\mu_s)$ design, where $\mu_s=\mu\cdot \frac{\binom{n-s}{t-s}}{\binom{\beta-s}{t-s}} .$
\end{lemma}
\begin{lemma}\label{lem2.2}\cite[Th. 8.1.6]{Huffman} Let $\mathcal{D}=(V,\mathcal{X})$ be an $(n,\beta,\mu)$-BIBD. Then the following hold:
\begin{itemize}
    \item[(a)] Every point in $V$ occurs in exactly $r=\frac{\mu(n-1)}{\beta-1}$ blocks (the number $r$ is called the replication number of $\mathcal{D}$).
    \item[(b)] $\mathcal{D}$ has exactly $b=\frac{nr}{\beta}$ blocks.  
    %\item[(c)] If $k\leq n-2,$ then there also exists a $(n,n-k,b-2r+\mu)$-BIBD.
\end{itemize}
    
\end{lemma}
\begin{definition}\cite[Def. 1.11]{Stinson}
    Let $\mathcal{D}=(V,\mathcal{X})$ be an $(n,\beta,\mu)$-BIBD, where $V=\{p_1,p_2,\ldots,p_n\}$
is the point set and
$\mathcal{X}=\{X_1,X_2,\ldots,X_b\}
$ is the collection of blocks.
The incidence matrix of $\mathcal{D}$ is an $n\times b$ binary matrix
$M=(m_{ij}),$ defined by
\begin{equation*}
     m_{ij}=\left\{ \begin{array}{cl}
     1 & \text{if } p_i\in X_j;  \\
     0 & \text{otherwise.}
 \end{array}\right.
 \end{equation*}

Note that the rows of $M$ correspond to the points of the design, the columns correspond to the blocks, and the $(i,j)$-entry indicates whether the point $p_i$ is contained in the block $X_j.$
\end{definition}
\begin{definition}\cite{Stinson}
    Let $\mathcal{D}=(V,\mathcal{X})$ be an $(n,\beta,\mu)$-BIBD with $\beta\leq n-2.$  The complementary design of $\mathcal{D},$ denoted by $\mathcal{D}^c,$  is obtained by replacing each block $B\in\mathcal{X}$ with its complement $V\setminus B.$ That is, \begin{equation*}
     \mathcal{D}^c=(V,\mathcal{X}^c),
 \end{equation*}  where  $\mathcal{X}^c=\{V\setminus B ~|~ B \in \mathcal{X}\}.$ 
\end{definition}

 \begin{theorem}\label{thm2.1}\cite[Th. 1.32]{Stinson} Let $\mathcal{D}=(V,\mathcal{X})$ be an $(n,\beta,\mu)$-BIBD with $\beta\leq n-2.$  The complementary design  $\mathcal{D}^c$ is an $(n,n-\beta,b-2r+\mu)$-BIBD, where $b$  denotes the number of blocks in $\mathcal{D}$  and $r$ is the replication number of $\mathcal{D}.$ 
     
 \end{theorem}
 
 % \begin{theorem}\label{thm2.1}\cite[Th. 1.32]{Stinson} Let $\mathcal{D}=(V,\mathcal{X})$ be an $(n,\beta,\mu)$-BIBD with $\beta\leq n-2.$  The complementary design of $\mathcal{D},$ denoted by $\mathcal{D}^c,$  is obtained by replacing each block $B\in\mathcal{X}$ with its complement $V\setminus B.$ That is, \begin{equation*}
 %     \mathcal{D}^c=(V,\mathcal{X}^c),
 % \end{equation*}  where  $\mathcal{X}^c=\{V\setminus B ~|~ B \in \mathcal{X}\}.$ Furthermore, $\mathcal{D}^c$ is an $(n,n-\beta,b-2r+\mu)$-BIBD, where $b$  denotes the number of blocks in $\mathcal{D}$  and $r$ is the replication number of $\mathcal{D}.$ 
     
 % \end{theorem}

 Ly and Soljanin \cite{Ly} established a link between the service rate region (SRR) of linear coded storage systems and combinatorial design theory. Using the design structure of dual codes, they derived bounds on the maximum achievable service rate of each data object and showed that, for systematic codes, the upper bound is achieved when the supports of minimum-weight dual codewords form a BIBD.

 %Ly and Soljanin \cite{Ly} demonstrated the relationship between SRR and combinatorial designs  and applied their framework to characterize the maximal achievable service rates of several code families. They derived general upper and lower bounds on the maximum achievable service rate of each data object using the dual code structure. They prove that these bounds are tight for previously studied non-systematic codes and establish that, for systematic codes, the upper bound is attained when the supports of minimum-weight dual codewords form a BIBD.
\begin{lemma}\cite[Th. 2]{Ly}\label{LyandSoljanin}
 Let $\mathcal{C}$ be an $[n,k]$  linear code over $\mathbb{F}_q$ with a systematic generator matrix $\mathcal{G}.$ 
 Then, for each $i \in[k],$ the maximum achievable service rate $\lambda_i^*$ for the $i$-th data symbol of $\mathcal{C}$ satisfies $\lambda_i^* \leq 1+\frac{n-1}{d^{\perp}-1}.$ In particular, if the supports of codewords of weight $d^{\perp}$ in $\mathcal{C}^\perp$ form a BIBD, then we have $\lambda_i^* = 1+\frac{n-1}{d^{\perp}-1}.$
\end{lemma} 
 %\begin{theorem}\label{thm2.2}\cite[Th. 8.4.1]{Huffman} \textbf{(Assmus–Mattson)}  Let $\mathcal{C}$ be a binary $[n,k,d]$ linear code, and let $\mathcal{C}^\perp$ be its dual code with  minimum Hamming distance $d^\perp.$ Let  $A_i^\perp=A_i(\mathcal{C}^\perp)$ denote the weight distribution of $\mathcal{C}^\perp$ for $0\le i\le n.$  Fix a positive integer $t<d$, and let $s$ be the number of integers $i$ satisfying $A_i^\perp\neq 0$ for $0<i\le n-t.$ If $s\le d-t,$ then the supports of the codewords of weight $i$ in $\mathcal{C}^\perp$ form a $t$-design whenever $A_i^\perp\neq 0$ and $d^\perp\le i\le n-t.$ \end{theorem}

 \section{Maximal achievable service rates of linear codes associated with BIBDs}\label{linear}
 %In this section, we investigate the maximal achievable service rates of linear codes associated with BIBDs. We derive explicit lower bounds on the maximum achievable service rates of the corresponding data symbols. 
 Extensive research has been done on obtaining the sufficient conditions under which the supports of codewords of a fixed weight in a code or in its dual code form a design (see the Assmus–Mattson theorem and its extensions \cite[Chap. 8]{Huffman}).  The following theorem establishes a lower bound on the maximum achievable service rate for each data symbol of a systematic linear code when the supports of codewords of a fixed weight in its dual code form a BIBD. 
  Further, when the design is formed by the supports of the minimum-weight codewords of the dual code, this bound is attained \cite{Ly}.
 \begin{theorem}\label{thm2.3}
 
  Let $\mathcal{C}$ be an $[n,k]$  linear code over $\mathbb{F}_q$ with a systematic generator matrix $\mathcal{G}.$ Suppose that $A_\omega^\perp \neq 0$ for some $\omega$ satisfying $d^\perp \leq \omega \leq n,$  and the supports of  the codewords of weight  $\omega$ in $\mathcal{C}^{\perp}$  form an $(n,\omega,\mu) $-BIBD for some $\mu >0.$ Then, for each $i \in[k],$ the maximum achievable service rate $\lambda_i^*$ for the $i$-th data symbol of $\mathcal{C}$ satisfies $\lambda_i^* \geq 1+\frac{n-1}{\omega-1}.$ In particular, if $\omega=d^\perp,$ then we have 
$\lambda_i^{*}= 1+\frac{n-1}{d^\perp-1}.$
 \end{theorem}
 \begin{proof} Fix $i\in[k].$ Since the generator matrix $\mathcal{G}$ is systematic, the $i$-th data symbol can always be accessed directly from its corresponding systematic node. As each server has unit capacity, this systematic node can contribute one unit to the service rate of the $i$-th data symbol.
By assumption, the supports of the codewords of weight $\omega$ in $\mathcal{C}^\perp$ form an $(n,\omega,\mu)$-BIBD. By Lemma \ref{lem2.2},   every point occurs in exactly
$r=\frac{\mu(n-1)}{\omega-1}$
blocks, and every pair of distinct points occurs in exactly $\mu$ blocks.
Consider the collection of all blocks containing the point $i.$ There are exactly $r$ such blocks. From Proposition \ref{Prop2.1}, each of these blocks corresponds to a recovery set for the $i$-th data symbol. Since every pair of distinct points occurs in exactly $\mu$ blocks, each node $j\neq i$ appears in precisely $\mu$ of these $r$ blocks. Assigning a load of $1/\mu$ to each block (recovery set) containing $i,$ a maximum of $\frac{r}{\mu}$ demands can be served by these available blocks. Combining this with the unit service rate supplied by the systematic node yields $\lambda_i^* \geq 1+\frac{n-1}{\omega-1}.$
Moreover, when $\omega=d^{\perp},$ by Lemma \ref{LyandSoljanin}, we get $
\lambda_i^{*}=1+\frac{n-1}{d^\perp-1},
$
which completes the proof. 

 \end{proof}
\begin{remark}
     To derive the lower bound on the maximum achievable service rate for a particular data symbol indexed by $i\in[k]$, it is not necessary that the generator matrix has to be in systematic form. The bound holds irrespective of whether the generator matrix is systematic, provided that the generator matrix contains a systematic server (node) corresponding to the $i$-th data symbol.  
\end{remark}

 A systematic generator matrix ensures that the smallest recovery set for each data symbol is a singleton set. The natural question is, how Theorem \ref{thm2.3} can be generalized for recovery sets of larger cardinality when the generator matrix is non-systematic. The following theorem addresses the case for binary codes when the supports of the codewords of a fixed weight in the dual code form a $t$-design, and the generator matrix of the code is not necessarily systematic.
\begin{theorem}\label{nonsymetriclinear}
    Let $\mathcal{C}$ be an $[n,k]$ binary linear code such that the supports of the codewords of weight $\omega$    in $\mathcal{C}^\perp$ form a $t-(n,\omega,\mu)$ design. For $i\in[k],$ if there exists a smallest recovery set of size $s$ for the $i$-th data symbol of $C$, where $1\leq s<t,$  then the maximal achievable service rate $\lambda_i^*$ of the $i$-th data symbol satisfies $\lambda_i^*\geq 1+\frac{n-s}{\omega-s}.$ In particular, if $\omega=d^{\perp},$ then we have $\lambda_i^*= 1+\frac{n-s}{d^\perp-s}.$
\end{theorem}
\begin{proof}
    Let $\mathcal{G}$ be a generator matrix of $\mathcal{C}$ and $S\subset [n]$ be a smallest  recovery set of size $s$  for the $i$-th data symbol  in $\mathcal{C}$, i.e., 
    \begin{equation*}
       \sum_{j\in S}\mathcal{G}^{(j)}=e_i.
   \end{equation*}
   Note that, for  $\mathbf{c}\in\mathcal{C}^\perp$ with $S\subset\mathcal{S}upp(\mathbf{c})$, we have  \begin{equation*}
       \sum_{j\in \mathcal{S}upp(\mathbf{c})}\mathcal{G}^{(j)}=\sum_{j_1\in S}\mathcal{G}^{(j_1)}+\sum_{j_2\in \mathcal{S}upp(\mathbf{c})\setminus S}\mathcal{G}^{(j_2)}=0,\end{equation*}
     which implies that  \begin{equation*}
           \sum_{j_2\in \mathcal{S}upp(\mathbf{c})\setminus S}\mathcal{G}^{(j_2)}=e_i.   
       ~~~~~~~~~~~~~~~~~~~~~~~~~~~~~~~~~~~~~~~\end{equation*}

       This implies that the set $\mathcal{S}upp(\mathbf{c})\setminus S$ is also a recovery set for the $i$-th data symbol.  By Lemma \ref{lem2.1}, it follows that the supports of the codewords of weight $\omega$  of $\mathcal{C}^\perp$ forms an $s-(n,\omega,\mu_s)$ design, where $\mu_s=\mu\cdot \frac{\binom{n-s}{t-s}}{\binom{\omega-s}{t-s}} .$ This implies that there are exactly $\mu_s$ codewords $\mathbf{c}\in\mathcal{C}^\perp$ with $S\subset\mathcal{S}upp(\mathbf{c})$, i.e., $\mu_s$ recovery sets of cardinality $\omega-|S|$ for the $i$-th data symbol that are disjoint from $S$ (throughout this paper, $|\cdot |$ denotes the cardinality function). To determine how many of these $\mu_s$ recovery sets contain any  node other than those in $S$, we note that 
       %it is necessary to analyze the number of blocks of the underlying design that contain $(s+1)$ nodes. Since
       $s+1\leq t$, and  by Lemma \ref{lem2.1}, we see that the supports of weight $\omega$ codewords of $\mathcal{C}^\perp$ also forms an $(s+1)-(n,\omega,\mu_{s+1})$ design, where $\mu_{s+1}=\mu\cdot \frac{\binom{n-s-1}{t-s-1}}{\binom{\omega-s-1}{t-s-1}} .$  Each subset of $[n]$ of size $(s+1)$ lies in the support of $\mu_{s+1}$ codewords of $\mathcal{C}^\perp.$ 
       Consequently, there are $\mu_s$ recovery sets for the $i$-th data symbol, with each node from $[n]\setminus S$ appearing in exactly $\mu_{s+1}$ of these recovery sets. By assigning a load of $\frac{1}{\mu_{s+1}}$ to each of these recovery sets, we can serve up to a total of $\frac{\mu_s}{\mu_{s+1}}$ requests. Combining this with the unit demand handled by the set $S$, we obtain a lower bound $\lambda_i^*\geq 1+\frac{\mu_s}{\mu_{s+1}}=1+\frac{n-s}{\omega-s}.$
 In particular, if $\omega=d^\perp,$ then by Theorem $1$ of Ly and Soljanin \cite{Ly}, we have $\lambda_i^*\leq 1+\frac{n-s}{d^\perp-s}.$ This completes the proof.  
\end{proof}
%If the recovery set is of the smallest size and the minimum weight codewords of the dual form $t$-design, then by Theorem $1$ of Ly and Soljanin \cite{Ly} and Theorem \ref{nonsymetriclinear}, it follows that $\lambda_i^*=1+\frac{n-s}{w-s}.$

In the next result, we consider a special case, when the smallest recovery set for a particular data symbol has size $2,$ and the supports of the codewords of a fixed weight of dual code form an SBIBD with $\mu=2.$ 
\begin{theorem}\label{thm2.4}
     Let $\mathcal{C}$ be an $[n,k]$ linear code over $\mathbb{F}_q$ such that the $\ell$-th data symbol has a smallest recovery set of size $2$. Suppose that $A_\omega^\perp \neq 0$ for some $\omega$ satisfying $d^\perp \leq \omega \leq n,$  and the supports of  the codewords of weight $\omega$ in $\mathcal{C}^{\perp}$  form an $(n,\omega,\mu=2) $-SBIBD. Then, the maximum achievable service rate for the $\ell$-th data symbol satisfies $\lambda_\ell^* \geq 3.$ 
%      %In particular, if $\omega=d^\perp,$ then
% $\lambda_i^{*}= 1+\frac{n-1}{d^\perp-1}$ for each $i\in[k].$
 \end{theorem}

\begin{proof}
Let $\{i,j\}$ be a smallest recovery set for the $\ell$-data symbol of size $2$.  Since the support of  codewords  of weight $\omega$  form an $(n,\omega,2)$-SBIBD, we have two codewords $\mathbf{c}, \mathbf{c'}\in \mathcal{C}^\perp$ such that $\mathcal{S}upp(\mathbf{c})\cap\mathcal{S}upp(\mathbf{c'})=\{i,j\}$. This implies that $\mathcal{S}upp(\mathbf{c})\setminus\{i,j\}$ and $\mathcal{S}upp(\mathbf{c'})\setminus\{i,j\}$ are two disjoint recovery sets for the $\ell$-th data symbol. Each of which can serve one request, that together with the unit request served by $\{i,j\}$, we have $\lambda_\ell^*\geq3$.
\end{proof}

\begin{example}
    The supports of all the nonzero codewords of the $[7, 3]$ binary simplex code form a $(7, 4, 2)$-SBIBD. Note that every nonzero codeword of the $[7,3]$ binary simplex code has weight $4$ and the dual of this code is the $[7,4]$ binary Hamming code. Therefore, in the  $[7, 4]$ binary Hamming code, the maximum achievable service rate for a data symbol is at least $3$, provided that there exists a smallest recovery set of size $2$ for that data symbol.
\end{example}

\section{Maximal achievable service rate of linear codes obtained using the incidence matrices of BIBDs}\label{incidence}

In this section, we consider linear codes constructed from the incidence matrices of BIBDs and show how the parameters of the underlying design can be used to establish lower bounds on the maximum achievable service rates for the data symbols of the corresponding linear code.

Let $\mathcal{D}=([n],\mathcal{X})$ be an $(n,\beta,\mu)$-BIBD with the incidence matrix $M.$ The columns of $M$ correspond to the incidence vectors of the blocks, and any two distinct rows of $M$ both contain $1$ in exactly $\mu$ columns.  Now consider a binary linear code $\mathcal{C}$ whose parity-check matrix is the transpose of the incidence matrix $M$, denoted by $M^\top.$  Thus, $\mathbf{c}\cdot M=0$ for all $\mathbf{c} \in \mathcal{C}.$ Note that all rows of $M^\top$ need not be linearly independent over $\mathbb{F}_2.$  Next, we observe that each row of $M^\top$ has weight $\beta$, and each column has weight $r=\frac{\mu(n-1)}{\beta-1}$. Note that all the rows of the matrix $M^\top$ (incidence vectors of blocks of $\mathcal{D}$) are  codewords of $\mathcal{C}^\perp.$  Further, it is evident that the maximum achievable service rate for any data symbol with respect to a systematic generator matrix of $\mathcal{C}$ is at least $1+\frac{r}{\mu}.$ 

When $\mu=1,$  we see that no pair of columns of $M$ contain $1$'s at the same two positions. Equivalently,  the intersection of the supports of any two rows of $M$ contains exactly one node. This leads to the following result.
\begin{theorem}\label{LDPC}
    Let $M$ be the incidence matrix of an $(n,\beta,\mu=1)$-BIBD, $\mathcal{D}=([n],\mathcal{X})$, where $b$ denotes the number of blocks in $\mathcal{D}.$ Let $k=b-rank(M)$ and  $\mathcal{C}$ be a $[b,k]$ binary linear code whose parity-check matrix is $M.$ Then, corresponding to a systematic generator matrix  of $\mathcal{C}$, we have $$\lambda_i^*\geq 1+\beta$$
for the $i$-th data symbol, where $i \in [k].$ 
\end{theorem}
\begin{proof}
Let us suppose that $\mathcal{X}=\{X_1,X_2,\ldots,X_b\}.$ Recall that the $j$-th column of $M$ is the incidence vector of the block $X_j.$  Fix $i \in [k]$. 
%For each data symbol indexed by $i\in[k]$, 
Since $|X_i|=\beta,$  exactly $\beta$ rows of $M$  have $1$ at the $i$-th coordinate position. Now, corresponding to each such row, we have a recovery set for the $i$-th data symbol. This means that there exist at least $\beta$ recovery sets for the $i$-th data symbol.
Since $\mu=1,$ every pair of distinct rows of $M$ contain $1$ in exactly one common coordinate position. This ensures that all these $\beta$ recovery sets are disjoint. Thus, the maximum service rate achievable for any data symbol of $\mathcal{C}$ is at least $1+\beta.$   
\end{proof}
\begin{remark}\label{shorten}
   Consider an $n\times b^\prime$ submatrix $M^\prime$ of the $n \times b$ incidence matrix $M$ of an $(n,\beta,\mu=1)$-BIBD. Let  $\mathcal{C}^\prime$ be the $[b^\prime, k^\prime]$ binary linear code obtained as the null space of the matrix $M',$ where $k^{\prime}=b^\prime-rank(M^\prime).$ Note that each column of $M^\prime$ has weight $\beta$, and any pair of distinct rows of $M'$ contain $1$ in at most one common coordinate position. Working as in proof of  Theorem \ref{LDPC}, it can be deduced that corresponding to a systematic generator matrix of $\mathcal{C}^\prime$, we have $\lambda_i^*\geq 1+\beta$ for all $i\in[k^\prime].$

% for each $i \in [b^\prime-rank(M^\prime)].$ Consider a $n\times b^\prime$ submatrix $M^\prime$ of the $n \times b$ incidence matrix $M$ of an $(n,\beta,\mu=1)$-BIBD, $\mathcal{D}=([n],\mathcal{X})$. The columns of $M'$ are incidence vectors corresponding to $b'$ blocks from $\mathcal{X}$. Define  $\mathcal{C}^\prime$ to be a $[b^\prime, k^\prime=b^\prime-rank(M^\prime)]$ linear code obtained as the null space of the matrix $M'$. Note that each column of $M^\prime$ has weight $\beta$, and any pair of distinct rows of $M'$ shares a 1 in exactly one coordinate position. Following a similar argument as in proof of Theorem \ref{LDPC}, it can be deduced that corresponding to a systematic generator matrix  of $\mathcal{C}^\prime$, we have $\lambda_i^*\geq 1+\beta$
% for each $i \in [k^\prime].$

\end{remark}
When the row weights and column weights of the parity-check matrix are very small compared to the length of the code, then the corresponding matrix is said to be a sparse parity-check matrix. For an $(n,\beta,\mu=1)$-BIBD, the incidence matrix $M$  has constant column-weight $\beta$ and constant row-weight $r$.  When $n$ is large compared to   $\beta$ and $r,$  then $M$ is a sparse matrix.   In fact, $M$ exhibits all the structural properties of a parity-check matrix for a regular LDPC code \cite{GalH}.  Consequently, the null space of $M$ defines a $(\beta; r)$-regular LDPC code of length $b$, known as a BIBD-LDPC code.
% The row weight $r$ and the column weight $\beta$ are very small compared to the length of the code, then the corresponding matrix is a sparse parity-check matrix. For a BIBD with the constraint $\mu=1$, its incidence matrix $M$ exhibits all the structural properties of a parity-check matrix for a regular LDPC code \cite{GalH}. The matrix $M$ has constant column-weight $\beta$, and constant row-weight $r$. Consequently, the null space of $M$ defines a $(\beta; r)$-regular low-density parity-check (LDPC) code, known as a BIBD-LDPC code, with length $b$. }
 The Tanner graph of a BIBD-LDPC code contains no cycles of length $4$ (In Tanner's framework, a $4$-cycle exists exactly when two variable nodes are connected to the same two check nodes, which means two columns of parity-check matrix share two common $1$'s.), meaning its girth (length of smallest cycle in Tanner graph) is at least $6$ \cite{Tanner}. The short cycles in the Tanner graph often indicate the dependencies among the recovery sets that reduce the number of disjoint recovery sets. Therefore, LDPC codes with large girth are often desirable.
Moreover, if the incidence matrix $M$ is a fully circulant square matrix ($n\times n$ matrix whose rows are cyclic shifts of the row above it), then the code with parity-check matrix $M$ is a cyclic LDPC code.

%the dual code as well as the LDPC code is cyclic.

\begin{remark}
    In an $(n,\beta,\mu=1)$-BIBD with a circulant square incidence matrix, the number of blocks equals the number of nodes, i.e., $b=n$. Hence, it is essentially an SBIBD. For an $(n,\beta,\mu=1)$-SBIBD, the incidence matrix $M$ adheres to the constant row-column constraints and can therefore be taken as a parity-check matrix of an LDPC code. Since $\mu=1,$ by Theorem 1.17 of \cite{Stinson}, we see that the dual design of this SBIBD is also an SBIBD with identical parameters. The matrix $M^\top$ likewise satisfies all the criteria for being a parity-check matrix of an LDPC code. Therefore, the maximum achievable service rate for each data symbol corresponding to a systematic generator matrix of both these LDPC codes is at least $1+\beta.$
\end{remark}
Now we will obtain explicit lower
bounds on the maximum achievable service rates for each data symbol of some classes of BIBD-LDPC codes that are constructed using BIBDs.

\begin{enumerate}
\item[(I)] Let $M$ be the incidence matrix of an $(n, 3,1)$-BIBD. MacKay and Davey \cite{STS} constructed LDPC codes of length $\frac{n(n-1)}{6}$  and dimension $\frac{n(n-1)}{6}- rank(M)$  whose parity-check matrix is the  matrix $M$.  By Theorem \ref{LDPC}, for these LDPC codes with a systematic generator matrix, we have $\lambda_i^*\geq 4$ for each $i\in[\frac{n(n-1)}{6}-rank(M) ].$

  \item[(II)] Bose \cite{Bose} constructed some BIBDs using the
method of symmetrically repeated differences, which are called Bose-BIBDs. Ammar {\etal}~\cite{Ammar} constructed several classes of LDPC codes with new parameters  using Bose-BIBDs. 

Now, we  consider some of these classes of LDPC codes and determine a lower bound on the maximum achievable service rate for each data symbol corresponding to a
systematic generator matrix of these LDPC codes.

\begin{enumerate}
\item[(a)]   Ammar {\etal}  \cite{Ammar} constructed Type-$1$ and Type-$2$ Class-I Bose  BIBDs with parameters $(12t+1,4,\mu=1)$ and $(20t+1,5,\mu=1),$ respectively. In both Type-$1$ and Type-$2$ Class-I Bose BIBDs, the incidence matrix $Q=\begin{bmatrix}
        Q_1&Q_2&\cdots&Q_t
    \end{bmatrix}$ consists of $t$ circulant  square matrices $Q_1,Q_2,\ldots,Q_t$ of order $12t+1$ and  $20t+1,$ respectively.
 The Class-I BIBD LDPC codes corresponding to Type-$1$ Class-I Bose BIBDs have parameters $[m(12t+1),(m-1)(12t+1)]$ and the Class-I BIBD LDPC codes corresponding to Type-$2$ Class-I Bose BIBDs  have parameters $[m(20t+1),(m-1)(20t+1)]$ which are obtained as the null space of the matrix 
$Q[m]=\begin{bmatrix}
        Q_1&Q_2&\cdots&Q_m
    \end{bmatrix}$ for some $1\le m\le t.$ 
    \\ Now, by Remark \ref{shorten}, we see that  $\lambda_i^*$  satisfies $\lambda_i^*\geq 5$ for  each $i$-th data symbol of the  Class-I BIBD LDPC codes corresponding to Type-$1$ Class-I Bose BIBD  with a systematic generator matrix.   
 Again, by Remark \ref{shorten}, we note that  $\lambda_i^*\geq 6$ for  each $i$-th data symbol of the  Class-I BIBD LDPC codes corresponding to Type-$2$ Class-I Bose BIBD  with a systematic generator matrix.

% The incidence matrix of a class-I
% Bose-BIBD consists of $t$ circulant square matrices. Among those $t$, if any $m$ are selected and a matrix is defined in the following way-
% \[H^{(1)}[m]=[Q_1,Q_2,\ldots,Q_m].\]

% Clearly, $H^{(1)}[m]$ is a parity-check matrix of a regular $(4, 4m)$-LDPC code with $\lambda_i^*\geq 5$ corresponding to a systematic generator matrix, for Type-$1$ design. Similarly, for the Type-$2$ design, the code obtained is $(5, 5m)$-LDPC code with $\lambda_i^*\geq 6$.
%    
\item[(b)] For any positive integer $t,$  Ammar {\etal}  \cite{Ammar} described $(12t+4, 4,1)$-Class-II Bose BIBDs based on Bose designs \cite{Bose}. They constructed  Class-II BIBD-LDPC codes of length $(3t+1)(4t+1) $ whose parity-check matrix  $M_2$  is the incidence matrix of a $(12t+4, 4,1)$-Class-II Bose BIBD.    By Theorem \ref{LDPC},  we see that for  Class-II BIBD-LDPC code  of length $(3t+1)(4t+1) $ and dimension $(3t+1)(4t+1)-rank(M_2)$ with a systematic generator matrix, $\lambda_i^*\geq 5$ holds for each data symbol.

% \textcolor{blue}{The Class-II Bose Designs, have the $ (12t+4)\times(3t+1)(4t+1) $ incidence matrix $Q$, with parameters $(n=12t+4,b=(3t+1)(4t+1), \beta=4, r=4t+1,\lambda=1)$. Therefore, for the Class-II BIBD-LDPC code, the maximum achievable service rate of any data symbol with respect to a systematic generator matrix of the LDPC is at least $1+\beta=5$.}
\item[(c)]  
For any positive integer $t,$ Ammar {\etal}  \cite{Ammar} constructed  Class-III BIBD-LDPC codes of length $(5t+1)(4t+1) $ whose parity-check matrix  $M_3$  is the incidence matrix of  a $(20t+5, 5,1)$-Class-III Bose BIBD.    By Theorem \ref{LDPC},  we see that for  Class-III BIBD-LDPC code  of length $(5t+1)(4t+1) $ and dimension $(5t+1)(4t+1)-rank(M_3)$ with a systematic generator matrix,  we have $\lambda_i^*\geq 6$ for each data symbol. 

% \textcolor{blue}{Class-III BIBD-LDPC Codes are also obtained by taking the incidence matrix of the BIBD with parameters $(v=20t+5,b=(5t+1)(4t+1),
% r=5t+1, \beta=5, \lambda=1)$. Then this Class-III BIBD-LDPC Code has $\lambda_i^*\geq 6$ corresponding to a systematic generator matrix.}

\end{enumerate}
\end{enumerate}
Although the bound given in Theorem \ref{LDPC} is comparatively small (due to only a subset of available recovery sets being considered) than the code length, it provides a guaranteed level of service rate irrespective of any specific request allocation strategy. One can find similar bounds for the existing constructions of LDPC codes, which are constructed using BIBDs and difference sets \cite{Johnson,Lan,  Vasic}.

\section{Maximal achievable service rate  of extended linear codes}\label{extended}
Let $\mathcal{C}$ be an $[n, k]$  linear code over $\mathbb{F}_q$ with the minimum Hamming distance $d$. The extended code $\overline{\mathcal{C}}$ of the code $\mathcal{C},$ is constructed by appending to each codeword $\mathbf{c}=(c_1,c_2,\ldots,c_n)\in \mathcal{C}$ an additional coordinate $c_{n+1}=-\sum\limits_{i=1}^{n}c_i.$ 
% Thus, the extended code $\overline{\mathcal{C}}$ of the code $\mathcal{C}$ is defined as  \begin{equation*}
%     \overline{\mathcal{C}}=\big\{(c_1,c_2,\ldots,c_n,-\sum\limits_{i=1}^{n}c_i )~|~ (c_1,c_2,\ldots, c_n)\in \mathcal{C}\big\}. 
% \end{equation*} 
The $(n+1)$-th coordinate of $\overline{\mathcal{C}}$ is called the overall parity-check coordinate. Note that the extended code $\overline{\mathcal{C}}$ is an $[n+1,k]$ linear code over $\mathbb{F}_q$ with the minimum Hamming distance $\overline{d}\in\{d,d+1\}.$ Next, let $\mathcal{G}=\begin{bmatrix}
   \mathcal{G}^{(1)}& \mathcal{G}^{(2)}&\cdots & \mathcal{G}^{(n)} 
\end{bmatrix}$ be a generator matrix of $\mathcal{C},$ where $\mathcal{G}^{(i)}$ denotes the $i$-th column of $\mathcal{G}$, $1\leq i \leq n.$  Then the generator matrix $\overline{\mathcal{G}}$ of the extended code $\overline{\mathcal{C}}$  is obtained by appending an additional column $\mathcal{G}^{(n+1)}$ to $\mathcal{G},$ where $\mathcal{G}^{(n+1)}=-\sum\limits_{i=1}^{n}\mathcal{G}^{(i)}.$ Thus, \begin{equation}\label{extendmatrix}
\overline{\mathcal{G}}=\begin{bmatrix} \mathcal{G}^{(1)}& \mathcal{G}^{(2)}&\cdots & \mathcal{G}^{(n)} & \mathcal{G}^{(n+1)} \end{bmatrix}.  
\end{equation}

Next, we observe that if the $(n+1)$-th coordinate of every codeword in the extended code $\overline{\mathcal{C}}$ is zero, then $\overline{\mathcal{C}}$ is simply the code $\mathcal{C}$  with an appended zero coordinate. Consequently, the service rate region of $\overline{\mathcal{C}}$ coincides with that of $\mathcal{C}$. Therefore, the study of the service rate region of the extended code is meaningful only when there exists at least one codeword in $\overline{\mathcal{C}}$ whose $(n+1)$-th coordinate is nonzero. If $\mathcal{C}$ is a binary linear code, then the extended code $\overline{\mathcal{C}}$ contains a codeword with a nonzero $(n+1)$-th coordinate if and only if $\mathcal{C}$ contains a codeword of odd Hamming weight. Further, we observe that every recovery set for a data symbol of $\mathcal{C}$ remains a recovery set for the corresponding data symbol of $\overline{\mathcal{C}}.$ Thus,  the extension preserves all recovery sets of the original code. Hence, for $i\in[k],$ the maximal achievable service rate $\overline{\lambda}_{i}^{*}$ for any $i$-th data symbol  corresponding to $\overline{ \mathcal{G}}$ is at least $\lambda_i^*$, which is the maximal achievable service rate for the $i$-th data symbol corresponding to $ \mathcal{G}.$

\begin{lemma}\label{lemma1}
  Let $\mathcal{C}$  be a linear code over $\mathbb{F}_q$ with a generator matrix $\mathcal{G}$  and let $\overline{\mathcal{C}}$  be its extended code with the generator matrix $\overline{\mathcal{G}}$.  Then, $\mathcal{R}^{min}(\mathcal{G})\subseteq\mathcal{R}^{min}(\overline{\mathcal{G}})$ and $\Lambda(\mathcal{R}^{\mathrm{min}}(\mathcal{G}),1)\subseteq\Lambda(\mathcal{R}^{\mathrm{min}}(\overline{\mathcal{G}}),1)$. 

\end{lemma}

\begin{remark}\label{rem3.1}
Suppose $w_H(\mathcal{G}^{(n+1)})=1 $ and $\mathcal{S}upp(\mathcal{G}^{(n+1)})=\{i\},$ where $i\in[k].$  Then the appended coordinate in $\overline{\mathcal{C}}$ forms a new recovery set for the $i$-th data symbol. Hence, the maximum achievable service rate for the $i$-th data symbol satisfies $\overline{\lambda}_{i}^*=\lambda_i^*+1.$ 
\end{remark}

By Lemma \ref{LyandSoljanin}, we observe that for a systematic generator matrix of an $[n,k]$ linear code $\mathcal{C},$ the associated achievable service rate for the $i$-th data symbol satisfies $\lambda_{i}\leq 1+\frac{n-1}{d^\perp-1}.$  Similarly, by Lemma \ref{LyandSoljanin}, for the extended code $\overline{\mathcal{C}}$ with the generator matrix $\overline{\mathcal{G}}$, we see that  the associated achievable service rate for the $i$-th data symbol satisfies $\overline{\lambda}_{i}\leq 1+\frac{n}{\overline{d}^\perp-1}$, where $\overline{d}^\perp$ denotes the minimum Hamming distance of the dual code $\overline{\mathcal{C}}^\perp.$ Comparing these bounds, we observe that the extension may increase the maximum achievable service rate. In particular, when $\overline{d}^{\perp} \le d^{\perp},$ we may have
$\lambda_i^* < \overline{\lambda}_{i}^{*}.$ We demonstrate this with the following example.
\begin{example}

Let $C$ be the $[23,12]$ binary  Golay code with a systematic generator matrix. Its dual code $C^\perp$ is a $[23,11]$ binary code with $d_H(C^\perp)=8$.
%By \textcolor{red}{Theorem ?? of \cite{??},} 
The supports of the minimum-weight codewords of $C^\perp$ form a $(23,8,56)$-BIBD. Therefore, Lemma \ref{LyandSoljanin}, the maximum achievable service rate for each data symbol satisfies $\lambda_i^*=1+\frac{22}{7}=4.14$ for $1\leq i \leq 12.$ Now consider the extended binary Golay code $\overline{C}$, which is a $[24,12]$ self-dual code with $d_H(\overline{C})=8.$ By Example 8.3.3 of \cite{Huffman}, the supports of the minimum-weight codewords of $\overline{C}^\perp$ form a $(24,8,77)$-BIBD. Applying Lemma \ref{LyandSoljanin} again, we obtain
$\overline{\lambda}_i^*=1+\frac{23}{7}=4.28$ for $1\leq i \leq 12.$  Note that $\lambda_i^* < \overline{\lambda}_i^*$ for  $1\leq i \leq 12$. This implies that the extended binary Golay code achieves a higher maximum service rate for every data symbol than the binary Golay code.
\end{example}

The following theorem provides a sufficient condition under which the minimum Hamming distance of the dual of the extended code does not exceed the minimum Hamming distance of the dual of the original code.
\begin{theorem}\label{thm3.1}
Let $\mathcal{C}$  be an $[n,k]$ linear code over $\mathbb{F}_q.$
%and let $d^\perp$ and $\overline{d}^\perp$ denote the minimum distances of $\mathcal{C}^\perp$  and $\overline{\mathcal{C}}^\perp,$ respectively. 
Assume that $\mathcal{C}^\perp$ contains a codeword having $s$ coordinates, all equal to the same element of  $\mathbb{F}_q^{*},$ where $s\geq n+1-d^\perp.$ Then we have $\overline{d}^\perp\leq d^\perp.$
\end{theorem}
\begin{proof} Let
$\mathbf{c}=(c_1,c_2,\ldots,c_n)\in \mathcal{C}^\perp$ be such that exactly $s$ coordinates of $\mathbf{c}$ are equal to $a\in\mathbb{F}_q^*.$ Note that  the all-one vector $\mathbf{1}=(1,1,\ldots,1)$ belongs to $\overline{\mathcal{C}}^{\perp}.$ This implies that $\mathbf{c}^\prime=(\mathbf{c},0)+(-a)\mathbf{1}=(c_1-a,c_2-a,\ldots,c_n-a,-a)\in \overline{\mathcal{C}}^{\perp}.$ 
Since $\mathbf{c}$  has exactly $s$ coordinates equal to $a,$ we have $w_H(\mathbf{c}^\prime)=n-s+1.$ Further, we observe that $\overline{d}^{\perp} \leq w_H(\mathbf{c}^\prime)=n-s+1\leq d^{\perp},$ which completes the proof. 
\end{proof}

\begin{corollary}\label{coro3.1}
Let $\mathcal{C}$ be a  %$[n,k,d]$ 
binary linear code of length $n.$
%and let $d^\perp$ and $\overline{d}^\perp$ denote the minimum distances of $\mathcal{C}^\perp$ and $\overline{\mathcal{C}}^\perp,$ respectively.
If $d^\perp \geq \frac{n+1}{2},$
then we have 
$\overline{d}^\perp \leq d^\perp.$
\end{corollary}

\begin{theorem}\label{thm3.5}
   Let $\mathcal{C}$ be an $[n,k]$  
   binary linear code. Suppose that the supports of the minimum weight codewords of $\mathcal{C}^\perp$  and  $\overline{\mathcal{C}}^\perp$ form BIBDs  with minimum weight $d^\perp>1$ and $\overline{d}^{\perp}>1,$ respectively.  If $\mathcal{C}^{\perp}$ contains a codeword of weight at least $n+1-d^\perp,$  then $\lambda_{i}^*<\overline{\lambda}_{i}^*$  for $1 \leq i \leq k.$ 
\end{theorem}  
\begin{proof}
   The desired result follows immediately from Lemma \ref{LyandSoljanin}  and Theorem \ref{thm3.1}. 
\end{proof}

The following theorem identifies a sufficient condition under which the service rates for two data symbols improve simultaneously in the extended code. 

\begin{theorem}\label{thm2}
Let $\mathcal{C}$  be an $[n,k]$ linear code over $\mathbb{F}_q$ with a systematic generator matrix, and let $\overline{\mathcal{C}}$  be its extended code.  Suppose that there exists a codeword $\mathbf{c}\in \mathcal{C}^\perp$ such that all nonzero coordinates of $\mathbf{c}$ are equal, $i\in\mathcal{S}upp(\mathbf{c})$ and $j\notin\mathcal{S}upp(\mathbf{c})$ for two distinct data symbol indices $i,j\in [k].$  Then 
\begin{equation*}
   \overline{\lambda}_{i}^*\geq 2, ~\overline{\lambda}_{j}^*\geq 2 \text{ and }(\overline{\lambda}_{i}+\overline{\lambda}_{j})^*\geq 4, \text{ where } (\overline{\lambda}_{i}+\overline{\lambda}_{j})^*=\max\{\overline{\lambda}_{i}+\overline{\lambda}_{j}\mid (\overline{\lambda}_1,\overline{\lambda}_2,\ldots,\overline{\lambda}_k)\in\Lambda(\overline{\mathcal{G}})\}.
\end{equation*} \end{theorem}

\begin{proof}
Since $\mathbf{c}\in \mathcal{C}^{\perp}$ and $i\in \mathcal{S}upp(\mathbf{c}),$ the set $R_i=\mathcal{S}upp(\mathbf{c})\setminus\{i\}$ is a recovery set for the $i$-th data symbol. Note that  $j\notin R_i.$ Next, we observe that all nonzero coordinates of $\mathbf{c}$  are equal and the all-one vector belongs to $\overline{\mathcal{C}}^\perp.$ This implies that  there exists a codeword $\mathbf{c}^{\prime}\in \overline{\mathcal{C}}^\perp$ whose support is
$ [n+1]\setminus\mathcal{S}upp(\mathbf{c})$
and  $j\in\mathcal{S}upp(\mathbf{c^\prime}).$ Therefore,
$ R_j=\mathcal{S}upp(\mathbf{c}^\prime)\setminus\{j\}$
is a recovery set for the $j$-th data symbol. Note that $R_i\cap R_j=\emptyset.$ Thus, each of the the $i$-th and $j$-th data symbols possesses at least two disjoint recovery sets. From this, the desired result follows. 
\end{proof}
\begin{remark}\label{rem3.2} Theorem \ref{thm2} is of particular interest for codes whose service rates satisfy $\lambda_i +\lambda_j \leq 3.$ For such codes, the extension may yield a strict improvement in the service rate. In particular, the binary extended Hamming code offers higher service rates than $\mathrm{Ham}(r,2),$ making it a more attractive choice for incomplete decoding.
\end{remark}

%~~~~~~~~~~~
%~~~~~~~~~~~
%~~~~~~~~~~~
The following theorem associates the design structures formed by the supports of codewords of weight $\omega$ and the supports of codewords of weight $n+1-\omega$ of $\mathcal{C}^{\perp}$ to the design structure induced by the corresponding weight codewords in $\overline{\mathcal{C}}^\perp.$
\begin{theorem}\label{thm3.3}
Let $\omega$ and $\omega_0$ be integers such that $3 \leq \omega \leq n-2$ and $\omega_0=n+1-\omega.$ Let $\mathcal{C}$ be a binary linear code of length $n$ with $A_\omega^\perp \neq 0$ and $A_{\omega_0}^\perp \neq 0.$    %(throughout this paper, $A_\omega^{\perp}$ denotes the number of codewords of weight $\omega$ in the dual code $\mathcal{C}^{\perp}$)
Suppose that the supports of codewords of   weight $\omega$  of $\mathcal{C}^{\perp}$ form a $t$-$(n,\omega,\mu_1)$ design and the supports of codewords of  weight $\omega_0$  of $\mathcal{C}^{\perp}$ form a $t$-$(n,\omega_0,\mu_2)$ design with $\mu_1(\omega_0-1)\binom{\omega_0-2}{t-2}=\mu_2(\omega-1)\binom{\omega-2}{t-2}$ for some $\mu_1, \mu_2 >0$ and $t\geq 2 .$   Then the supports of codewords of  weight $\omega$  of  $\overline{\mathcal{C}}^{\perp}$ form an $(n+1,\omega,\mu)$-BIBD  with $\mu=\mu_2\left(\frac{n-1}{\omega_0-1}\right) \left( \frac{n-\omega_0}{w_0}\right)\frac{\binom{n-2}{t-2}}{\binom{\omega_0-2}{t-2}}.$

\end{theorem}
\begin{proof}
    Since the supports of codewords of weight $\omega$  of $\mathcal{C}^\perp$  form a $t-(n,\omega,\mu_1)$ design and the supports of codewords of  weight $\omega_0$  of $\mathcal{C}^\perp$  form a $t-(n,\omega_0,\mu_2)$ design, it follows from Lemma \ref{lem2.1} that they also form an $(n,\omega,\widehat{\mu}_1)$-BIBD  and an $(n,\omega_0,\widehat{\mu}_2)$-BIBD, respectively, where  $\widehat{\mu}_1=\mu_1\cdot \frac{\binom{n-2}{t-2}}{\binom{\omega-2}{t-2}}$ and $\widehat{\mu}_2=\mu_2\cdot \frac{\binom{n-2}{t-2}}{\binom{\omega_0-2}{t-2}}$. 

    Let $\mathcal{D}_1=([n],\mathcal{X}_1)$ be the $(n,\omega,\widehat{\mu}_1)$-BIBD formed by the supports of all codewords of weight $\omega$ in $\mathcal{C}^\perp.$  By Lemma \ref{lem2.2}, we see that each point in $[n]$ occurs in exactly $r_1=\frac{\widehat{\mu}_1(n-1)}{\omega-1}$ blocks and $\mathcal{D}_1$ has exactly $b_1=\frac{nr_1}{\omega}$ blocks. Next, let $\mathcal{D}_2=([n],\mathcal{X}_2)$ be the $(n,\omega_0,\widehat{\mu}_2)$-BIBD formed by the supports of all codewords of weight $\omega_0$ in $\mathcal{C}^\perp.$  Again, by Lemma \ref{lem2.2}, we see that each point in $[n]$ occurs in exactly $r_2=\frac{\widehat{\mu}_2(n-1)}{\omega_0-1}$ blocks and $\mathcal{D}_2$ has exactly $b_2=\frac{nr_2}{\omega_0}$ blocks. By Theorem \ref{thm2.1}, we see that $\mathcal{D}_2^c$ is an $(n,n-\omega_0,b_2-2r_2+\widehat{\mu}_2)$-BIBD with the replication number $r_2^c=b_2-r_2.$ 
    
    Now consider the set $\mathcal{M}=\mathcal{M}_1\cup \mathcal{M}_2,$
  where 
        $\mathcal{M}_1=\{(\mathbf{c},0) ~|~ \mathbf{c}\in \mathcal{C}^{\perp} \text{ and } w_H(\mathbf{c})=\omega\} $ and $\mathcal{M}_2= \{(\textbf{1}+\mathbf{c}^{\prime},1) ~|~ \mathbf{c}^\prime \in \mathcal{C}^{\perp} \text{ and } w_H(\mathbf{c}^\prime)=\omega_0\}.$ It is easy to see that $\mathcal{M}\subseteq \overline{\mathcal{C}}^{\perp}$ and $w_H(\mathbf{d})=\omega $ for all $\mathbf{d} \in \mathcal{M}.$ 
%It is easy to see that every vector in $\mathcal{M}$  belongs to $\overline{\mathcal{C}}^{\perp}$  and has Hamming weight $\omega.$ 
Next, we observe  that adding the $(n+1)$-th  node to all the blocks of $\mathcal{D}_2^c$ yields the supports of the  set $\mathcal{M}_2.$ 
Now consider the set \begin{equation*}
    \widetilde{\mathcal{D}}=\{\mathcal{S}upp(\mathbf{d}) ~|~ \mathbf{d}\in \mathcal{M}\}.
\end{equation*} 
We claim that $\mathcal{X}=([n+1],\widetilde{\mathcal{D}})$ is an $(n+1,\omega, \mu )$-BIBD with $\mu=\mu_2\left(\frac{n-1}{\omega_0-1}\right) \left( \frac{n-\omega_0}{w_0}\right)\frac{\binom{n-2}{t-2}}{\binom{\omega_0-2}{t-2}}$. To prove this, it suffices to show that every point $i\in [n+1]$
 occurs in exactly the same number of blocks of $\mathcal{X}$, and every pair $\{i,j\}$ of distinct points of $[n+1]$ is contained in exactly $\mu$ blocks of $\mathcal{X}.$ 
 
 To prove that every point $i\in [n+1]$
 occurs in the same number of blocks of $\mathcal{X},$  we first note that the $(n+1)$-th node is contained in exactly $b_2$ blocks. Further, we observe that each $i\in[n]$ appears in exactly $r_1+r_2^c$ blocks.  Now using the fact that $\mu_1(\omega_0-1)\binom{\omega_0-2}{t-2}=\mu_2(\omega-1)\binom{\omega-2}{t-2}, $ it is easy to see that $b_2= r_1+r_2^c.$ This implies that each point $i \in[n+1]$ appears in exactly $b_2$ blocks. 
 
 Next, to prove that every pair $\{i,j\}$ of distinct points of $[n+1]$ is contained in exactly $\mu$ 
 blocks of $\mathcal{X},$ we observe that 
the pair $\{i,n+1\}$ is contained in exactly $b_2-r_2$ blocks and the pair $\{i,j\}$ is contained in exactly $\widehat{\mu}_1+b_2-2r_2+\widehat{\mu}_2=b_2-r_2$ blocks for $i,j\in[n].$  This proves our claim. 
\end{proof}
% \begin{remark}
%   \textcolor{red}{ The working of Theorem \ref{thm3.3} requires the codewords of weight $\omega$ and $\omega_0$ in  $\mathcal{C}^\perp$ to form $t$-designs. For the same, sufficient conditions are given in Theorem \ref{thm2.2}, and one of the assumptions leads to  $t+1\leq \omega\leq n-t.$}
% \end{remark}
\begin{corollary}\label{coro3.2}
Let $\omega=\frac{n+1}{2}$, and let $\mathcal{C}$ be a  binary linear code of length $n$ with $A_\omega^\perp \neq 0.$ Suppose that the supports of all codewords of weight $\omega$ in $\mathcal{C}^{\perp}$  form a $t$-$(n,\omega,\mu_1) $  design  for some $\mu_1 >0$ and $t\geq 2.$  Then the supports of the weight $\omega$ codewords of $\overline{\mathcal{C}}^{\perp}$ forms an $(n+1,\omega,\mu)$-BIBD  with $\mu= 2\mu_1\left(\frac{n-1}{n+1}\right) \frac{\binom{n-2}{t-2}}{\binom{\omega-2}{t-2}}.$
\end{corollary}

\begin{theorem}\label{thm3.4}
    Let $\mathcal{C}$  be an $[n,k]$   binary linear code with a systematic generator matrix, and suppose that $A_\omega^\perp \neq 0,$ where $\omega=\frac{n+1}{2}.$ If the supports of all codewords of weight $\omega$ in $\mathcal{C}^\perp$ form an $(n,\omega,\mu)$-BIBD for some $\mu>0.$ Then for $1\leq i \leq k,$ we have
    $\lambda_i^*\geq 1+ \frac{n-1}{\omega-1}=3$  and 
         $\overline{\lambda}_{i}^*\geq 3+\frac{2}{n-1} .$

\end{theorem}
%\begin{proof}
 %   \end{proof}

\begin{remark}\label{rem5.4}
   Let $C$ be an $[n,k]$ binary linear code with a systematic generator matrix. By  Lemma \ref{LyandSoljanin}, we see that $\lambda_i^* \leq 1+\frac{n-1}{d^\perp-1} $ and $\overline{\lambda}_i^*\leq 1+ \frac{n}{\overline{d}^\perp-1}$ for each $i \in[k].$   
   %where $d^\perp$ and $\overline{d}^\perp$ denote the minimum distance of the dual code $C^\perp$ and $\overline{C}^\perp,$ respectively.
   Now, if $\omega=d^{\perp}$ in  Theorem \ref{thm3.4}, then  we get   $\lambda_i^*=3$  and 
    $\overline{\lambda}_{i}^*=3+\frac{2}{n-1}$ for each $i \in[k].$  
\end{remark}

\begin{example}
    \noindent \textbf{Application to Extended Ham($r, 2$)}:  Let $\mathcal{C}$ be the binary Hamming code with parameters $[2^r-1,2^r-1-r].$ Recall that $\mathcal{C}^\perp$ is the Simplex code with parameters $[2^r-1,r]$ and minimum Hamming distance $2^{r-1}.$  From  Section $4.2$ of \cite{Ding}, we see that the supports of the minimum weight codewords of $\mathcal{C}^{\perp}$ form a $(2^r-1,2^{r-1},2^{r-2})$-BIBD. Then by Remark \ref{rem5.4}, we have $\lambda_i^*=3$ and $\overline{\lambda}_i^*=3+\frac{1}{2^{r-1}-1}$ for each $i \in[2^r-1-r].$ Note that $\lambda_i^*<\overline{\lambda}_i^*  $  for each $i \in[2^r-1-r].$

\end{example}

\section{Cyclic codes and their connection with  SBIBDs}\label{cyclicsection}

In this section, we first recall some basic properties of cyclic codes and their duals. We then establish a connection between cyclic codes and SBIBDs, which provides a combinatorial framework for studying the maximal achievable service rates associated with these codes.

% \begin{definition}\cite[Def. 7.1.1]{Ling} \textbf{Cyclic code:}
  A linear code $\mathcal{C}$ of length $n$ over $\mathbb{F}_q$ is called a cyclic code if it is invariant under cyclic shifts, that is, $(c_{n},c_1,c_2,\ldots,\\c_{n-1})\in \mathcal{C}$    whenever
$(c_1,c_2,\ldots,c_{n})\in \mathcal{C}.$

We  identify a vector 
 $(c_1,c_2,\ldots,c_n)$
  with the polynomial
 $c_1+c_2x+c_3x^2+\cdots+c_{n}x^{n-1}
\in \mathbb{F}_q[x]/\langle x^n-1\rangle.$
%In view of this, by Theorem 7.2.1 of \cite{Ling}, we see that 
Then a cyclic code of length $n$ over $\mathbb{F}_q$ can be viewed as an ideal of the quotient ring
$\mathbb{F}_q[x]/\langle x^n-1\rangle.$ 
Note that $\mathbb{F}_q[x]/\langle x^n-1\rangle$ is a principal ideal ring. Further, by Theorem 4.2.1 of \cite{Huffman}, every cyclic code $\mathcal{C}$ is generated by a unique monic divisor $g(x)$  of $x^n-1,$ \textit{i.e.,} $\mathcal{C}=\langle g(x)\rangle.$ The polynomial $g(x)$ is called the generator polynomial of $\mathcal{C}$.
%andwe have $\mathcal{C}=\langle g(x)\rangle.$
If $\deg(g(x))=n-k,$ then $\mathcal{C}$ is an $[n,k]$ cyclic code.
The polynomial 
$h(x)=\frac{x^n-1}{g(x)}(=h_1+h_2x+\cdots+h_{k+1}x^{k})$
is called the parity-check polynomial of $\mathcal{C}.$  A polynomial $c(x)\in\mathbb{F}_q[x]/\langle x^n-1\rangle$ is a codeword in $\mathcal{C}$  if and only if $h(x)c(x)\equiv 0 \pmod{x^n-1}.$ Now, let us define  $h^*(x)=x^{k}h(x^{-1}),$ (called the reciprocal polynomial of $h(x)$).  By Theorem 4.2.7 of \cite{Huffman},  we see that $\mathcal{C}^\perp$ is also cyclic, and its generator polynomial is 
$h^\perp(x)=h^*(x)h_1^{-1}.$ 
%That is, we have$\mathcal{C}^\perp=\langle h^\perp(x)\rangle.$
%Throughout this paper, for a subset $A=\{a_1,a_2,\ldots,a_s\}\subseteq \mathbb{Z}_n$and an integer $a,$  we define$A+a=\{(a_i+a)\bmod n : a_i\in A\}.$ 

% For convenience, we label the coordinate positions of a cyclic code by the elements of \textcolor{blue}{the ring of integers modulo $n$, $\mathbb{Z}_n=\{0, 1,\ldots,n-1\}$, identifying the $i$-th coordinate position in a codeword with the $(i-1)$-th element of $\mathbb{Z}_n.$}
% Throughout this paper, for a subset $J\subseteq \mathbb{Z}_n$ and an integer $a,$ the notation
% $J+a=\{(j+a)\bmod n : j\in J\}$
% denotes the translation (or cyclic shift) of $J$ by $a$ modulo $n.$ 
% Let $\mathcal{C}$ be an $[n,k]$ cyclic code over $\mathbb{F}_q$ with a systematic generator matrix $\mathcal{G}.$ Now if $\mathbf{c}=(c_1,c_2,\ldots,c_{n})$ is a  codeword in the dual code $\mathcal{C}^{\perp}$ with $X=\mathcal{S}upp(\mathbf{c}),$  then $\mathbf{c}^{\prime}=(c_{n},c_1,\ldots,c_{n-1})\in \mathcal{C}^{\perp}$ and support of the codeword $\mathbf{c}^{\prime}$ is $X+1.$ This implies that there exists codewords with supports $X+1,X+2,\ldots,X+n-1$ in the dual code $\mathcal{C}^{\perp}.$ \textcolor{blue}{ Note that the supports of codewords and their cyclic shifts are subsets of $\mathbb{Z}_n$.}
For convenience, we label the coordinate positions of a cyclic code  integer modulo $n,$ that is, we will  identify the $i$-th coordinate position in a codeword with the $(i-1)$-th element of $\mathbb{Z}_n.$
Throughout this paper, for a subset $J\subseteq \mathbb{Z}_n$ and an integer $a,$ the notation
$J+a=\{(j+a)\bmod n : j\in J\}$
denotes the translation (or cyclic shift) of $J$ by $a$ modulo $n.$ 
Let $\mathcal{C}$ be an $[n,k]$ cyclic code over $\mathbb{F}_q$ with a systematic generator matrix $\mathcal{G}.$ Now if $\mathbf{c}=(c_1,c_2,\ldots,c_{n})$ is a  codeword in $\mathcal{C}^{\perp}$ with $X=\mathcal{S}upp(\mathbf{c}),$  then $\mathbf{c}^{\prime}=(c_{n},c_1,\ldots,c_{n-1})\in \mathcal{C}^{\perp}$ and support of the codeword $\mathbf{c}^{\prime}$ is $X+1.$ This implies that there exists codewords with supports $X+1,X+2,\ldots,X+n-1$ in the dual code $\mathcal{C}^{\perp}.$  Note that the supports of codewords and their cyclic shifts are subsets of $\mathbb{Z}_n$.

Let $\mathcal{C}$ be a cyclic code. The following theorem shows how the existence of an SBIBD corresponding to a codeword in $\mathcal{C}^\perp$  can be used to derive a lower bound on the maximum achievable service rate of $\mathcal{C}.$
\begin{theorem}\label{thm4.1}
  Let $\mathcal{C}$ be an $[n,k]$ cyclic code over $\mathbb{F}_q$ with a systematic generator matrix $\mathcal{G}.$ Let  $\mathbf{c}=(c_1,c_2,\ldots,c_{n})$ be a  codeword in the dual code $\mathcal{C}^{\perp}$ with $X=\mathcal{S}upp(\mathbf{c})\subset\mathbb{Z}_n$ and $|X|=\omega.$ Suppose that $\mathcal{D}=(V,\mathcal{X})$ is an $(n,\omega,\mu)$-SBIBD, where $V=\mathbb{Z}_n,$ and  $\mathcal{X}=\{X+1,X+2,\ldots,X+n-1\}.$ Then the maximum  achievable service rate $\lambda_i^*$ for the $i$-th data symbol satisfies $\lambda_i^* \geq 1+\frac{n-1}{\omega-1}.$ In particular, if $\omega=d^\perp,$ then we have 
$\lambda_i^{*}= 1+\frac{n-1}{d^\perp-1}$ for each $i\in[k].$
\end{theorem}
\begin{proof}
   Since $\mathcal{D}=(V,\mathcal{X})$  is an $(n,\omega,\mu )$  SBIBD, it follows that the number of blocks $b=n$  and the replication number is $r=\frac{\mu(n-1)}{\omega-1}.$  Thus, each point of $V$ occurs in exactly $r$ blocks of $\mathcal{X}$. Each of these $r$ blocks corresponds to the support of a dual codeword, i.e., to a recovery set. Now using the same approach as in the proof of Theorem \ref{thm2.3}, the desired result follows immediately.  
\end{proof}

\begin{definition}\cite[Def. 3.1]{Stinson}
 A subset
$D=\{d_1,d_2,\ldots,d_\beta\}\subseteq \mathbb{Z}_n$
is called an $(n,\beta,\mu)$-difference set in $\mathbb{Z}_n$ if every nonzero element of $\mathbb{Z}_n$ can be expressed as a difference
$d_i-d_j \pmod n,$ 
in exactly $\mu$ ordered ways, where  $d_i,d_j\in D$ with  $d_i\neq d_j.$
\end{definition}

 By Theorem 3.8 of \cite{Stinson}, we note that if the set $D$ is an $(n,\beta,\mu)$-difference set in $\mathbb{Z}_n$, then $\mathcal{D}=(V,\mathcal{X})$ is an $(n,\beta,\mu)$ SBIBD, where $V=\mathbb{Z}_n$ and $\mathcal{X}=\{D,D+1,D+2,\ldots,D+n-1\}.$

\begin{corollary}\label{coro4.1}
      Let $\mathcal{C}$ be an $[n,k]$ cyclic code over $\mathbb{F}_q$ with a systematic generator matrix $\mathcal{G}.$ Let  $\mathbf{c}=(c_1,c_2,\ldots,c_{n})$ be a  codeword in the dual code $\mathcal{C}^{\perp}$ with $X=\mathcal{S}upp(\mathbf{c})$ and $|X|=\omega.$ Suppose that $X$ is an $(n,\omega,\mu)$-difference set in $\mathbb{Z}_n.$  Then the maximum  achievable service rate for the $i$-th data symbol satisfies $\lambda_i^* \geq 1+\frac{n-1}{\omega-1}$ for each $i\in[k].$ In particular, if $\omega=d^\perp,$ then  we have 
$\lambda_i^{*}= 1+\frac{n-1}{d^\perp-1}$ for each $i\in[k].$
\end{corollary}
\begin{proof}
   The desired result follows from Theorem 3.8 of \cite{Stinson} and Theorem \ref{thm4.1}. 
\end{proof}

\begin{remark}
 Theorem \ref{thm4.1} and Corollary \ref{coro4.1} remain valid for constacyclic codes also. Indeed, the proofs of Theorem \ref{thm4.1} and Corollary \ref{coro4.1}  rely only on the existence of a systematic generator matrix and the SBIBD structure induced by the support of a suitable codeword in the dual code. Therefore, Theorem \ref{thm4.1} and Corollary \ref{coro4.1}   extend naturally to the constacyclic setting.  
\end{remark}
\begin{remark}\label{rem4.1}
It is easier to check whether the support of the coefficient vector of the parity-check polynomial $h(x)$ forms a difference set. If it forms a difference set, then the result of Corollary \ref{coro4.1} follows immediately.
    \end{remark}
\begin{example}
    Let $C$ be the cyclic $[11,6]$ ternary Golay code with a systematic generator matrix 
    
    \begin{equation*}G=
        \begin{bmatrix}
        1&0&0&0&0&0&0& 1 &1& 1& 1\\
        0& 1 &0&0&0&0&1&0&1&2&2\\
        0&0&1&0&0&0&1&1&0&1&2\\
        0&0&0&1&0&0&1&2&1&0&1\\
        0&0&0&0&1&0&1&2&2&1&0\\0&0&0&0&0&1&1&1&2&2&1
    \end{bmatrix}.
    \end{equation*} Note that $C^{\perp}$ is an $[11,5]$-code over $\mathbb{F}_3$ with the minimum Hamming distance $6.$ It is easy to see that $\mathbf{c}=(1,0,1,1,1,0,0,\\0,1,0,1)\in C^{\perp}$ and $w_H(\mathbf{c})=6.$   Further, we observe that $\mathcal{S}upp(\mathbf{c})=\{0,2,3,4,8,10\}$ is an $(11,6,3)$-difference set in $\mathbb{Z}_{11}.$  Now, by Corollary \ref{coro4.1},  we see that the maximum  achievable service rate for the $i$-th data symbol is given by $\lambda_i^*=3$ for $1\leq i \leq 6.$
\end{example}

%\{1,3,4,5,9,11\}
Given an $(n,\omega,\mu)$-difference set, one can construct cyclic codes of length $n$ associated with it.  Similar constructions have been considered in \cite{Graham,Weldon}. Now, we briefly outline the construction method given in \cite{Graham,Weldon}.

\vspace{2mm}
\noindent\textbf{Construction (A)}: Let $D$ be an $(n,\omega,\mu)$-difference set in $\mathbb{Z}_n$. Define the polynomial
$a(x)=a_0+a_1x+\cdots+a_{n-1}x^{n-1}\in \mathbb{F}_q[x],$
where $a_i\in \mathbb{F}_q^*$ if $i\in D$ and $a_i=0$ otherwise. Thus, $\mathcal{S}upp(a(x))=D.$ (Throughout the paper, $\mathcal{S}upp(b(x))$ denotes the support of the coefficient vector of the polynomial $b(x)\in \mathbb{F}_q[x]$.) Let 
$h(x)=\gcd(a(x),x^n-1),$
and define
$g(x)=\frac{x^n-1}{h(x)}
=g_0+g_1x+\cdots+g_kx^k.$
Let
$g^{*}(x)=x^k g(x^{-1})$ be the
the reciprocal polynomial of $g(x).$ Consider the cyclic code $C$ of length $n$ over $\mathbb{F}_q$ generated by
$g^{\perp}(x)=g_0^{-1}g^{*}(x).$ We call the cyclic code $C$ obtained through the above construction the cyclic code associated with the difference set $D.$
 Now, by Corollary \ref{coro4.1}, corresponding to a systematic generator matrix of $C,$  the maximum  achievable service rate for the $i$-th data symbol satisfies
$
\lambda_i^{*}\geq 1+\frac{n-1}{\omega-1}$ for $ i\in [k].$
 We present below several families of cyclic codes constructed from difference sets below. Using the SBIBDs associated with these difference sets, we obtain explicit lower bounds on the maximum achievable service rates for each data symbol.

\begin{itemize}
    \item[(I)] Let $p$ be a prime and $s$ be a positive integer, and let us take
$n=p^{2s}+p^{s}+1.$
By Theorem 3.26 of \cite{Stinson}, there exists an $(n,p^{s}+1,1)$-difference set $D$ in $\mathbb{Z}_n.$ Let
$D=\{d_1,d_2,\ldots,d_{p^{s}+1}\}.$
Now, corresponding to $D,$  consider the polynomial
$\theta(x)=\sum_{i=1}^{p^{s}+1}x^{d_i}.$
Let
$h(x)=\gcd(\theta(x),x^n-1),$
and define
$g(x)=\frac{x^n-1}{h(x)}.$
By Theorem 2 of  Graham and MacWilliams \cite{Graham}, the cyclic code generated by $h(x)$ over $\mathbb{F}_p$ has parameter
$\left[n,\left(\frac{p+1}{2}\right)^s+1\right]$ with minimum Hamming distance $p^s+1.$
Now let $C_1$ be the cyclic code over $\mathbb{F}_p$ generated by $g^{\perp}(x).$ Note that $\dim (C_1)=n-1-\binom{p+1}{2}^s. $   Weldon \cite{Weldon} showed that this class of cyclic codes is almost as powerful as BCH codes and considerably simpler to implement.

Now, by Corollary \ref{coro4.1}, corresponding to a systematic generator matrix of $C_1$, the maximum  achievable service rate for the $i$-th data symbol is given by    $\lambda_i^{*}= p^{s}+2$ for each $i\in\left[n-1-\binom{p+1}{2}^s\right].$

 \item[(II)] Let $p\equiv 3\pmod 4$ be a prime. Consider the set, $QR(p)=\{a^2: a \in \mathbb{Z}_p\setminus \{0\}\},$ of quadratic residues modulo $p$. 
 By Theorem 3.21 of \cite{Stinson}, we see that  the set $QR(p)$  form a $(p,\frac{p-1}{2}, \frac{p-3}{4})$-difference set in $\mathbb{Z}_p.$  Consider the cyclic code $C_2$ obtained by Construction (A) associated with the difference set $QR(p).$  Now, by Corollary \ref{coro4.1}, corresponding to a systematic generator matrix of $C_2$, the maximum achievable service rate for the $i$-th data symbol is given by  $\lambda_i^{*}\geq  1+\frac{2(p-1)}{p-3}$ for each $i.$

    \item[(III)] Let $p=4t^2+1$ be  a prime,  where $t$ is an odd integer. Consider the set,  $ \widetilde{QR}(p)=\{z^4:z \in \mathbb{Z}_p\setminus \{0\} \}$, of quartic residues modulo $p$. By Theorem 3.23 of \cite{Stinson}, we see that  the set $\widetilde{QR}(p)$ form a $(4t^2+1,t^2,\frac{t^2-1}{4})$-difference set in $\mathbb{Z}_p.$ Consider the cyclic code $C_3$ obtained by Construction (A) associated with the difference set $\widetilde{QR}(p).$  Again, by Corollary \ref{coro4.1}, corresponding to a systematic generator matrix  of $C_3$, the maximum  achievable service rate for the $i$-th data symbol is given by $\lambda_i^{*}\geq  1+\frac{4t^2}{t^2-1}$ for each $i.$ 
\end{itemize}

\section{Conclusion and future work}\label{conclusion}
This paper examined lower bounds on the maximum achievable service rates for data symbols in various families of linear codes associated with combinatorial structures such as $t$-designs, difference sets, and balanced incomplete block designs.  We established a lower bound on the maximum achievable service rate for each data symbol of a systematic linear code when the supports of codewords of a fixed weight in the dual code form a BIBD, as well as for the non-systematic binary case when the supports of codewords of a fixed weight in the dual code form a $t$-design. Furthermore, we analyzed how the parameters of the  design can be used to obtain lower bounds on the maximum achievable service rates for the data symbols of the corresponding linear codes, which are obtained through the incidence matrix of the design. We also explored the maximum achievable service rate of systematic extended linear codes and 
systematic cyclic codes. It would be interesting to derive explicit bounds on maximum achievable service rates of locally recoverable codes. Another worthwhile direction will be to investigate how the recovery-set structure of a data symbol in newly constructed codes differs from that in the original codes when the former are derived from the latter. Additionally, one may explore the properties of a generator matrix that yields a larger service region than the standard systematic generator matrix of a given code. Finding a generator matrix that gives the optimal SRR for a code still remains an open problem.

\section*{Acknowledgments}
The first author is grateful to the University Grants Commission (UGC), Government of India, for the financial support received under the UGC NET-SRF scheme. The second author gratefully acknowledges the financial support provided by the Anusandhan National Research Foundation (ANRF), Government of India, under the NPDF scheme through Grant No. PDF/2025/000550.

{}

\begin{thebibliography}{}
   
  \bibitem{MAktas}{Aktas, M.,  Anderson, S. E.,  Johnston, A.,  Joshi, G.,  Kadhe, S.,   Matthews, G. L., Mayer, C., and Soljanin, E.}: \newblock{On the service capacity region of accessing erasure coded content},\newblock{\it{ Proc. 55th Annu. Allerton Conf. Commun., Control, Comput.}}, pp. 17-24 (2017).

\bibitem{Aktas}{Aktas, M.S., Joshi, G.,  Kadhe, S.,   Kazemi, F. and  Soljanin, E.}: \newblock{Service rate region: A new aspect of coded distributed system design},\newblock{\it{ IEEE Trans. Inf. Theory}}, 67(12), pp. 7940-7963 (2021).

\bibitem{alfarano2024service}{Alfarano, G.N., Kılıç, A. B., Ravagnani, A., and Soljanin, E.}: \newblock{The service rate region polytope}, \newblock{\it{SIAM J. Appl. Algebra Geom.}}, 8(3), pp. 553-582 (2024).
  
 \bibitem{LRC}{Alfarano, G.N., Kılıç, A.B. and Ravagnani, A.}: \newblock{Service Aspects of LRC and Batch Codes}, \newblock{\it{IEEE BITS the Information Theory Magazine}}, 3(4), pp. 17-27 (2023). 
 \bibitem{Alfarano}{ Alfarano, G.N., Ravagnani, A. and Soljanin, E.}: \newblock{Dual-code bounds on multiple concurrent (local) data recovery},\newblock{\it{IEEE Int. Symp. Inf. Theory (ISIT)}}, pp. 2613-2618 (2022). 
  
 \bibitem{Ammar}{Ammar, B., Honary, B., Kou, Y., Xu, J. and Lin, S.}:  \newblock{Construction of low-density parity-check codes based on balanced incomplete block designs}, \newblock{\it{IEEE Trans. Inf. Theory}}, 50(6), pp. 1257-1269 (2004).
  
    \bibitem{Bose}{Bose, R. C.}:\newblock{ On the construction of balanced incomplete block designs},\newblock{\it{ Ann. Eugenics }},
9, pp. 353-399  (1939).
 
 \bibitem{HamPC}{Choudhary, P. and Bhaintwal, M.}: \newblock{The Service Rate Region of Hamming Codes},\newblock{\it{ IEEE Trans. Inf. Theory}},\newblock{DOI 10.1109/TIT.2026.3709105} (2025).
 
 
 
 \bibitem{oval}{Di Giusto, A., Ravagnani, A. and Soljanin, E.}: \newblock{ The Oval Strikes Back},\newblock{\it{  arXiv:2601.16628}} (2026).
 
\bibitem{Dimakis}{ Dimakis, A. G.,  Godfrey, P. B.,  Wu, Y.,  Wainwright, M. J. and  Ramchandran, K.}: \newblock{Network coding for distributed storage systems},\newblock{\it{ IEEE Trans. Inf. Theory}}, 56(9), pp. 4539-4551 (2010).
 
\bibitem{Ding}{Ding, C. and Tang, C.}: \newblock{The linear codes of t-designs held in the Reed-Muller and Simplex codes},\newblock{\it{Cryptogr. Commun.}}, 13(6), pp. 927-949 (2021).
  
\bibitem{GalH}{Gallager, R.}:  \newblock{Low-density parity-check codes}, \newblock{\it{IRE Trans. Inform. Theory}}, 8(1), pp. 21-28 (1962).
  
 \bibitem{Gardner}{  Gardner, K.,  Zbarsky, S., Doroudi, S.,  Harchol-Balter, M. and  Hyytia, E.}: \newblock{Reducing latency via redundant requests: Exact analysis},\newblock{\it{ ACM SIGMETRICS Perform. Eval. Rev.}}, 43(1), pp. 347-360 (2015).
  
  \bibitem{Gopalan}{Gopalan, P., Huang, C.,  Simitci, H. and  Yekhanin, S.}: \newblock{On the locality of codeword symbols},\newblock{\it{ IEEE Trans. Inf. Theory}}, 58(11), pp. 6925-6934 (2012).
  
\bibitem{Graham}{Graham, R.L. and MacWilliams, J.}:\newblock{ On the Number of Information Symbols in Difference‐Set Cyclic Codes},\newblock{ \it{ Bell System Technical Journal}}, 45(7), pp. 1057-1070 (1966).
 
 \bibitem{Huang}{  Huang, C.,  Chen, M. and  Li, J.}: \newblock{Pyramid codes: Flexible schemes to trade space for access efficiency in reliable data storage systems},\newblock{\it{ ACM Trans. Storage}}, 9(1), pp. 1-28 (2013).
  
 \bibitem{Huffman}{ Huffman, W. C. and Pless, V.}: \newblock{\it{Fundamentals of error-correcting codes}}, \newblock{Cambridge Univ. Press, Cambridge, New York, USA} (2003).
 
\bibitem{Johnson}{Johnson, S.J. and Weller, S.R.}: \newblock{Construction of low-density parity-check codes from Kirkman triple systems}, \newblock{\it{ IEEE Global Telecommunications Conference (Cat. No. 01CH37270)}}, 2, pp. 970-974 (2001). 
   
   \bibitem{Joshi}{  Joshi, G.,  Liu, Y. and  Soljanin, E.}: \newblock{Coding for fast content download},\newblock{\it{ in Proc. 50th Annu. Allerton Conf. Commun., Control, Comput.}},  pp. 326-333 (2012).
 
 \bibitem{Combit}{Kazemi, F., Karimi, E., Soljanin, E. and Sprintson, A.}:  \newblock{A combinatorial view of the service rates of codes problem, its equivalence to fractional matching and its connection with batch codes}, \newblock{\it{ IEEE Int. Symp. Inf. Theory (ISIT)}}, pp. 646-651 (2020).
 
  \bibitem{Geo}{Kazemi, F., Kurz, S. and Soljanin, E.}: \newblock{A geometric view of the service rates of codes problem and its application to the service rate of the first order Reed-Muller codes},\newblock{\it{IEEE Int. Symp. Inf. Theory (ISIT)}}, pp. 66-71 (2020).
 
\bibitem{Lan}{Lan, L., Tai, Y.Y., Lin, S., Memari, B. and Honary, B.}: \newblock{New constructions of quasi-cyclic LDPC codes based on special classes of BIBD's for the AWGN and binary erasure channels}, \newblock{\it{IEEE Trans. Commun.}}, 56(1), pp. 39-48 (2008).
 
\bibitem{Liang}{  Liang, G.  and Kozat, U. C.}: \newblock{Fast cloud: Pushing the envelope on delay performance of cloud storage with coding},\newblock{\it{ IEEE/ACM Trans. Netw.}}, 22(6), pp. 2012–2025 (2013).
  
 \bibitem{Ling}{ Ling, S. and Xing, C.}: \newblock{\it{Coding theory: A first course}}, \newblock{Cambridge Univ. Press, Cambridge, New York, USA} (2004).
     
     \bibitem{Ly}{ Ly, H. and Soljanin, E.}: \newblock{Maximal achievable service rates of codes and connections to combinatorial designs},\newblock{\it{ 	 in Proc. 61st Annu. Allerton Conf. Commun., Control, Comput.}} (2025). %arXiv:2506.16983 [cs.IT]
 
\bibitem{LyH}{  Ly, H. and Soljanin, E.}: \newblock{Service rate regions of MDS codes and fractional matchings in quasi-uniform hypergraphs},\newblock{\it{  IEEE Trans. Inf. Theory}}, 72(4), pp. 2144-2161 (2026).
    
 \bibitem{LyHS}{ Ly, H., Soljanin, E. and Lalitha, V.}: \newblock{On the service rate region of Reed–Muller codes},\newblock{\it{  IEEE Trans. Inf. Theory}}, 72(8), pp. 5525-5542 (2026).%\newblock{DOI 10.1109/TIT.2026.3681406} 
 

  
  \bibitem{STS}{MacKay, D.J. and Davey, M.C.}:  \newblock{Evaluation of Gallager codes for short block length and high rate applications}, \newblock{\it{In Codes, Systems, and Graphical Models}}, pp. 113-130 (2001). 
    
    \bibitem{Noori}{  Noori, M.,  Soljanin, E. and  Ardakani, M.}: \newblock{On storage allocation for maximum service rate in distributed storage systems},\newblock{\it{ 	Proc. IEEE Int. Symp. Inf. Theory (ISIT)}}, pp. 240-244 (2016).
 

\bibitem{Stinson}{Stinson, D. R.}: \newblock{\it{Combinatorial Designs - Constructions and Analysis}}, \newblock{Springer, New York} (2004). 
  
\bibitem{Tanner}{Tanner, R.}: \newblock{A recursive approach to low complexity codes}, \newblock{\it{IEEE Trans. Inf. Theory}}, 27(5), pp. 533-547 (1981).
  
 \bibitem{Vasic}{Vasic, B. and Milenkovic, O.}: \newblock{Combinatorial constructions of low-density parity-check codes for iterative decoding},\newblock{\it{ IEEE Trans. Inf. Theory}}, 50(6), pp. 1156-1176 (2004).
 
\bibitem{Weldon}{Weldon, E.J.}:\newblock{ Difference-set cyclic codes},\newblock{\it{ The Bell System Technical Journal}}, 45(7), pp. 1045-1055 (1966).
 
 % \textcolor{blue}{\bibitem{Weller}{Johnson, S. and Weller, S.}: \newblock{Quasi-cyclic LDPC codes from difference families},\newblock{ \it{ 
    \end{thebibliography}
\end{document}